\documentclass[orivec]{llncs}
\pdfoutput=1
\def\techreport{}

\usepackage[utf8]{inputenc}
\usepackage{paralist}
\usepackage{enumitem}
\usepackage{comment}

\usepackage{graphicx}
\graphicspath{{./graphics/}}
\usepackage{wrapfig}

\usepackage{tikz}
\usetikzlibrary{arrows}

\usepackage{pgfplots}

\usepackage{color}
\usepackage[referable]{threeparttablex}

\usepackage{dcolumn}
\newcolumntype{d}[1]{D{.}{.}{#1}}

\usepackage{amsmath}
\usepackage{amssymb}
\usepackage{thm-restate}
\usepackage{mathtools}
\usepackage{oubraces}
\usepackage{stmaryrd}
\SetSymbolFont{stmry}{bold}{U}{stmry}{m}{n} 

\usepackage[ruled,linesnumbered,vlined]{algorithm2e}
\DontPrintSemicolon
\let\oldnl\nl
\newcommand{\nonl}{\renewcommand{\nl}{\let\nl\oldnl}}

\usepackage{colonequals}
\newcommand{\ass}{\colonequals}

\newcommand{\appendixreftr}[1]{%
\ifdefined\techreport%
#1%
\else%
our technical report~\cite{techreport}%
\fi%
}

\newcommand{\appendixref}[1]{%
\ifdefined\techreport%
#1%
\fi%
}

\newcommand{\RSM}       {\mathcal{R}}
\newcommand{\CRSM}      {\mathcal{R^{\parallel}}}
\newcommand{\Module}    {\mathcal{M}}
\newcommand{\Automaton} {\mathcal{C}}

\newcommand{\WPDS}      {\mathcal{P}}
\newcommand{\CPDS}      {\mathcal{P}^{\parallel}}
\newcommand{\APN}       {\mathcal{A}^{\parallel}}

\newcommand{\swappush}{\mathsf{sp}}

\newcommand{\Nodes}{N}
\newcommand{\In}  {\mathit{In}}   
\newcommand{\En}  {\mathit{En}}   
\newcommand{\Ex}  {\mathit{Ex}}   
\newcommand{\Call}{\mathit{Call}} 
\newcommand{\Ret} {\mathit{Ret}}  

\newcommand{\cnd}[2]{\tup{#1,#2}} 
\newcommand{\rnd}[2]{\tup{#1,#2}} 

\newcommand{\post}{\mathit{post}}

\newcommand{\Comps}{\Pi}
\newcommand{\Runs}{\Lambda}

\newcommand{\conf} [2]{\tup{#1,#2}} 
\newcommand{\lconf}[2]{\tup{#1,#2}}    
\newcommand{\gconf}[3]{\tup{#1,#2, #3}}    

\newcommand{\sem}[1]{\llbracket #1 \rrbracket}

\newcommand{\trans}[1][]{\xrightarrow{#1}}
\newcommand{\Trans}[1][]{\xRightarrow{#1}}

\newcommand{\reach}[2][]{\xrightarrow{#1}\mathrel{\vphantom{\to}^{#2}}}
\newcommand{\Reach}[2][]{\xRightarrow{#1}\mathrel{\vphantom{\to}^{#2}}}

\newcommand{\state}[2]{\tup{#1, #2}} 

\newcommand{\lang}{\mathcal{L}}

\newcommand{\Weights}{D}
\newcommand{\zero}{\overline{0}}
\newcommand{\one} {\overline{1}}
\newcommand{\semiring}{\tup{\Weights, \oplus, \otimes, \zero, \one}}
\newcommand{\sleq}{\sqsubseteq}

\newcommand{\Height}{H}

\newcommand{\booleanSemiring} {\tup{ \tup{0,1}, \lor, \land, 0, 1 }}

\newcommand{\rsmwt} {w}          
\newcommand{\autwt} {\ell}       
\newcommand{\compwt}{\otimes}    
\newcommand{\runwt} {\otimes}    
\newcommand{\Compswt}{\bigoplus} 
\newcommand{\rsmdist}{d}         

\newcommand{\WL}     {\mathsf{WL}}

\newcommand{\summary}{\mathsf{sum}}

\SetKwFunction{Relax}{Relax}
\newcommand{\ConfDistBounded}{\mathtt{ConfDist}}

\newcommand{\Marks}{\mathbb{M}}
\newcommand{\freshMark}{\widehat{m}}

\newcommand{\Outgoing}{\mathsf{out}}

\newcommand{\eps}{\varepsilon}

\newcommand{\Map}{\mapsto}

\DeclareMathOperator{\poly}{poly}

\newcommand{\set}  [1]{\{#1\}}
\newcommand{\tup}  [1]{\langle#1\rangle}

\newcommand{\ov}{\overline}

\makeatletter
\DeclareFontFamily{OMX}{MnSymbolE}{}
\DeclareSymbolFont{MnLargeSymbols}{OMX}{MnSymbolE}{m}{n}
\SetSymbolFont{MnLargeSymbols}{bold}{OMX}{MnSymbolE}{b}{n}
\DeclareFontShape{OMX}{MnSymbolE}{m}{n}{
    <-6>  MnSymbolE5
   <6-7>  MnSymbolE6
   <7-8>  MnSymbolE7
   <8-9>  MnSymbolE8
   <9-10> MnSymbolE9
  <10-12> MnSymbolE10
  <12->   MnSymbolE12
}{}
\DeclareFontShape{OMX}{MnSymbolE}{b}{n}{
    <-6>  MnSymbolE-Bold5
   <6-7>  MnSymbolE-Bold6
   <7-8>  MnSymbolE-Bold7
   <8-9>  MnSymbolE-Bold8
   <9-10> MnSymbolE-Bold9
  <10-12> MnSymbolE-Bold10
  <12->   MnSymbolE-Bold12
}{}

\let\llangle\@undefined
\let\rrangle\@undefined
\DeclareMathDelimiter{\llangle}{\mathopen}%
                     {MnLargeSymbols}{'164}{MnLargeSymbols}{'164}
\DeclareMathDelimiter{\rrangle}{\mathclose}%
                     {MnLargeSymbols}{'171}{MnLargeSymbols}{'171}
\makeatother

\newcommand{\myparagraph}[1]{\smallskip\noindent{\bf #1}}

\begin{document}

\newcommand{\grants}{%
  This research was supported in part by the Austrian Science Fund (FWF) under
  grants \mbox{S11402-N23}, \mbox{S11407-N23}, \mbox{P23499-N23}, and
  \mbox{Z211-N23}, and by the European Research Council (ERC) under grant
  \mbox{279307}.%
}

\title{Faster Algorithms for\\ Weighted Recursive State Machines}

\author{Krishnendu Chatterjee\inst{1}\and Bernhard Kragl\inst{1}\and\\ Samarth Mishra\inst{2}\and Andreas Pavlogiannis\inst{1}}
\institute{IST Austria, Klosterneuburg, Austria \and IIT Bombay, Mumbai, India}

\maketitle

\begin{abstract}
  Pushdown systems (PDSs) and recursive state machines (RSMs), which are
  linearly equivalent, are standard models for interprocedural analysis.
  Yet RSMs are more convenient as they (a)~explicitly model function calls and
  returns, and (b)~specify many natural parameters for algorithmic analysis,
  e.g., the number of entries and exits.
  We consider a general framework where RSM transitions are labeled from a
  semiring and path properties are algebraic with semiring operations, which can
  model, e.g., interprocedural reachability and dataflow analysis problems.

  Our main contributions are new algorithms for several fundamental problems. 
  As compared to a direct translation of RSMs to PDSs and the best-known
  existing bounds of PDSs, our analysis algorithm improves the complexity for
  finite-height semirings (that subsumes reachability and standard dataflow
  properties). 
  We further consider the problem of extracting distance values from the
  representation structures computed by our algorithm, and give efficient
  algorithms that distinguish the complexity of a one-time preprocessing from
  the complexity of each individual query.
  Another advantage of our algorithm is that our improvements carry over to the
  concurrent setting, where we improve the best-known complexity for the
  context-bounded analysis of concurrent RSMs.
  Finally, we provide a prototype implementation that gives a significant
  speed-up on several benchmarks from the SLAM/SDV project.
\end{abstract}

\section{Introduction}
\label{sec:introduction}

\myparagraph{Interprocedural analysis.}
One of the classical algorithmic analysis problems in programming languages is
the interprocedural analysis. 
The problem is at the heart of several key applications, ranging from alias
analysis, to data dependencies (modification and reference side effect), to
constant propagation, to live and use
analysis~\cite{Reps95,Sagiv96,Callahan86,Grove93,Land91,Knoop96,Cousot77,Giegerich81,Knoop92}.
In seminal works~\cite{Reps95,Sagiv96} it was shown that a large class of
interprocedural dataflow analysis problems can be solved in polynomial time. 

\myparagraph{Models for interprocedural analysis.}
Two standard models for interprocedural analysis are \emph{pushdown systems} (or
finite automata with stacks) and \emph{recursive state machines 
 (RSMs)}~\cite{AlurBEGRY05,HBMC:Procedural}. 
An RSM is a formal model for control flow graphs of programs with recursion. 
We consider RSMs that consist of modules, one for each method or function that
has a number of entry nodes and a number of exit nodes, and each module contains
boxes that represent calls to other modules. 
A special case of RSMs with a single entry and a single exit node for every
module (\emph{SESE RSMs}, aka \emph{supergraph} in~\cite{Reps95}) has also been
considered.
While pushdown systems and RSMs are linearly equivalent (i.e., there is a linear
translation from one model to the other and vice versa), there are two distinct
advantages of RSMs.
First, the model of RSMs closely resembles the problems of programming languages
with explicit function calls and returns, and hence even its special cases such
as SESE RSMs has been considered to model many applications.
Second, the model of RSMs provides many parameters, such as the number of entry
and exit nodes, and the number of modules, and better algorithms can be
developed by considering that some parameters are small.
Typically the SESE RSMs can model data-independent interprocedural analysis,
whereas general RSMs can model data dependency as well.
For most applications, the number of entries and exits of a module, usually
represents the input parameters of the module. 

\myparagraph{Semiring framework.}
We consider a general framework to express computation properties of RSMs where
the transitions of an RSM are labeled from a semiring. 
The labels are referred to as weights.
A {\em computation} of an RSM executes transitions between configurations
consisting of a node (representing the current control state) and a stack of
boxes (representing the current calling context).
To express properties of interest we need to define how to assign weights to
computations, i.e., to accumulate weights \emph{along} a computation, and how to
assign weights to sets of computations, i.e., to combine weights \emph{across} a
set of computations.
The weight of a given computation is the semiring product of the weights on the
individual transitions of the computation, and the weight of a given set of
computations is the semiring plus of the weights of the individual computations
in the set. 
For example, (i)~with the Boolean semiring (with semiring product as AND, and
semiring plus as OR) we express the reachability property; (ii)~with a Dataflow
semiring we can express problems from dataflow analysis.
One class of such problems is given by the IFDS/IDE
framework~\cite{Reps95,Sagiv96} that considers the propagation of dataflow facts
along distributive dataflow functions (note that the IFDS/IDE framework only
considers SESE RSMs).
Hence the large and important class of dataflow analysis problems that can be
expressed in the IFDS/IDE framework can also be expressed in our framework.
Pushdown systems with semiring weights have also been extensively considered in
the literature~\cite{Reps05,LalRB05,RepsLK07,LalR08}.

\myparagraph{Problems considered.}
We consider the following basic \emph{distance problems}.
\begin{compactitem}
\item {\em Configuration distance.}
Given a set of {\em source} configurations and a {\em target} configuration, the
{\em configuration distance} is the weight of the set of computations that start
at some source configuration and end in the target configuration. 
In the {\em configuration distance problem} the input is a set of source
configurations and the output is the configuration distance to all reachable
configurations.

\item {\em Superconfiguration distance.}
We also consider a related problem of {\em superconfiguration distance}.
A {\em superconfiguration} represents a sequence of modules, rather than a sequence
of invocations. 
Intuitively, it does not consider the sequence of function calls, but only which 
functions were invoked.
This is a coarser notion than configurations and allows for fast overapproximation.
The superconfiguration distance problem is then similar to the configuration distance
problem, with configurations replaced by superconfigurations.

\item {\em Node distance.}
Given a set of source configurations and a target node, the {\em node distance}
is the weight of the set of computations that start at some source configuration
and end in a configuration with the target node (with arbitrary stack). 
In the {\em node distance problem} the input is a set of source configurations
and the output is the node distance to all reachable nodes.
\end{compactitem}

\myparagraph{Symbolic representation.}
A core ingredient for solving distance problems is the symbolic representation
of sets of RSM configurations and their manipulation.
Given a symbolic representation of the set of initial configurations, we provide
a two step approach to solve the distance problems.
In step one we compute a symbolic representation of the set of all
configurations reachable from the initial configurations.
Furthermore, the transitions in the representation are annotated with
appropriate semiring weights to capture the various distances described above.
In step two we query the computed representation for the required distances.
Thus we make the important distinction between the complexity of a
\emph{one-time} preprocessing and the complexity of every \emph{individual
  query}.

\myparagraph{Concurrent RSMs.} 
While reachability is the most basic property, the study of pushdown systems and
RSMs with the semiring framework is the fundamental quantitative extension of
the basic problem. 
An orthogonal fundamental extension is to study the reachability property in a
{\em concurrent} setting, rather than the sequential setting.
However, the reachability problem in concurrent RSMs (equivalently concurrent
pushdown systems) is undecidable~\cite{Ramalingam00}. 
A very relevant problem to study in the concurrent setting is to consider
context-bounded reachability, where at most $k$ context switches are
allowed.
The context-bounded reachability problem is both decidable~\cite{QadeerR05} and
practically relevant~\cite{MusuvathiQ07,MusuvathiQBBNN08}.

\myparagraph{Previous results.}
Many previous results have been established for pushdown systems, and the
translation of RSMs to pushdown systems implies that similar results carry over
to RSMs as well.
We describe the most relevant previous results with respect to our results.
For an RSM $\RSM$, let $|\RSM|$ denote its size, $\theta_e$ and $\theta_x$ the
maximum number of entries and exits, respectively, and $f$ the number of modules. 
The existing results for weighted pushdown systems over semirings of height
$\Height$~\cite{Schwoon02,Reps05} along with the linear translation of RSMs to
pushdown systems~\cite{AlurBEGRY05} gives an
$O(\Height \cdot |\RSM| \cdot \theta_e \cdot \theta_x\cdot f )$-time algorithm for the
configuration and node distance problems for RSMs.
The previous results for context-bounded reachability of concurrent pushdown systems~\cite{QadeerR05}
applied to concurrent RSMs gives the following complexity bound: 
$O(|\CRSM|^5 \cdot \theta^{||~5}_x\cdot n^k\cdot |G|^k)$, where 
$|\CRSM|$ is the size of the concurrent RSM,   
$\theta^{||}_x$ is the number of exit nodes, 
$n$ is the number of component RSMs,
$G$ is the global part of the concurrent RSM,
and $k$ is the bound on the number of context switches.

\begin{table}[t]
\centering
\scriptsize
\renewcommand{\arraystretch}{1.2}
\begin{tabular}{|l|ll|ll|}
\hline
&\multicolumn{2}{c|}{\bf Sequential}&\multicolumn{2}{c|}{\bf Concurrent}\\
\hline
\hline
Existing & $H\cdot |\RSM|\cdot \theta_e\cdot \theta_x\cdot f$ & \cite{Schwoon02,Reps05} & $|\CRSM|^5 \cdot \theta^{||~5}_x\cdot n^k\cdot |G|^k$ & \cite{QadeerR05}\\
\hline
Our result & $\Height\cdot (|\RSM|\cdot \theta_e + |\Call|\cdot \theta_e\cdot \theta_x)$ & [Theorem~\ref{them:finite-height}] & $|\CRSM|\cdot \theta_e^{||}\cdot  \theta_x^{||}\cdot  n^k\cdot |G|^{k+2}$ & [Theorem~\ref{them:concurrent}]\\
\hline
\end{tabular}
\caption{Asymptotic time complexity of computing configuration automata.}
\label{tab:intro1}
\end{table}


\begin{table*}[t]
\begin{ThreePartTable}
\centerline{
\scriptsize
\renewcommand{\arraystretch}{1.2}
\begin{tabular}{|c|c|c|c|c|c|c|}
\hline
\multicolumn{5}{|c|}{Semiring} & \multicolumn{2}{c|}{RSM}\\
\hline
General & Boolean & Constant size & \multicolumn{2}{c|}{Size $|D|$\tnote{*}} & \multicolumn{2}{c|}{Sparse\tnote{$\dagger$}}\\
\hline
Query   & Query   & Query         & Preprocess & Query                       & Preprocess & Query\\
\hline
\hline
$n\cdot \theta_e^2$ &
$|\RSM|\cdot \theta_e \cdot\frac{n}{\log n}$ &
$n\cdot \frac{\theta_e^2}{\log \theta_e}$ &
$|\RSM|\cdot \theta_e^{1+\eps \cdot \log |D|}$ &
$n\cdot \frac{\theta_e^2}{\eps^2\cdot \log^2 \theta_e}$ &
$|\RSM|\cdot \theta_e^{\omega-1}\cdot x$ &
$n\cdot \left \lceil \frac{\theta_e^2}{\log x} \right \rceil$\\
\hline
\end{tabular}
}
\caption{Asymptotic time complexity of answering a configuration/superconfiguration
  distance query of size $n$. 
  Preprocess time refers to additional preprocessing after the configuration
  automaton is constructed.}
\label{tab:intro2}
\begin{tablenotes}
\item[*] For any fixed $\eps>0$.
\item[$\dagger$] In a sparse RSM every module only calls a constant number of other
  modules, and the result applies only to superconfiguration distances. 
  The parameter $x$ has to satisfy $x = O(\poly(|\RSM|))$, and
  $\omega$ is the smallest constant required for multiplying two square matrices
  of size $m\times m$ in time $O(m^\omega)$ (currently $\omega\simeq 2.372$).
\end{tablenotes}
\end{ThreePartTable}
\end{table*}


\myparagraph{Our contributions.}
Our main contributions are as follows:

\begin{compactenum}
\item[1.] 
  \emph{Finite-height semirings.}
  We present an algorithm for computing configuration and node distance problems
  for RSMs over semirings with finite height $\Height$ with running time
  $O(\Height\cdot (|\RSM|\cdot \theta_e + |\Call|\cdot \theta_e\cdot
  \theta_x))$, where $|\Call|$ is the number of call nodes.
  The algorithm we present constructs the symbolic representations from which the
  distances can be extracted.
  Thus our algorithm improves the current best-known algorithms by a factor of
  $\Omega((|\RSM|\cdot f)/(\theta_x + |\Call|))$ (Table~\ref{tab:intro1}) 
  (also see Remark~\ref{rem:finite-height} for details).

\item[2.] 
  \emph{Distance queries.}
  Once a symbolic representation is constructed, it can be used for extracting
  distances.
  We present algorithms which given a configuration query of size $n$, return
  the distance in $O(n \cdot \theta_e^2)$ time. 
  Furthermore, we present several improvements for the case when the semiring
  has a small domain.
  Finally, we show that when the RSM has a sparse call graph, we can obtain a
  range of tradeoffs between preprocessing and querying times.
  Our results on distance queries are summarized in Table~\ref{tab:intro2}.

\item[3.]
 \emph{Concurrent RSMs.} 
 For the context-bounded reachability of concurrent RSMs we present an algorithm with time bound
 $O(|\CRSM|\cdot \theta_e^{||}\cdot  \theta_x^{||}\cdot  n^k\cdot |G|^{k+2})$. 
 Thus our algorithm significantly improves the current best-known algorithm 
 (Table~\ref{tab:intro1}).

\item[4.] 
  \emph{Experimental results.}
  We experiment with a basic prototype implementation for our algorithms.
  Our implementation is an explicit (rather than symbolic) one.
  We compare our implementation with jMoped~\cite{jmoped}, which is a leading
  and mature tool for weighted pushdown systems, on several real-world
  benchmarks coming from the SLAM/SDV project~\cite{BallBLKL10,BallR00}.
  We consider the basic reachability property (representative for finite-height
  semirings) for the sequential setting.
  Our experimental results show that our algorithm provides significant
  improvements on the benchmarks compared to jMoped.
\end{compactenum}

\myparagraph{Technical contribution.} 
The main technical contributions are as follows:
\begin{compactitem}
\item We show how to combine (i)~the notion of {\em configuration automata} as a
  \emph{symbolic} representation structure for sets of configurations, and
  (ii)~entry-to-exit \emph{summaries} to avoid redundant computations, and
  obtain an efficient dynamic programming algorithm for various distance
  problems in RSMs over finite-height semirings.

\item Configuration and superconfiguration distances are extracted using graph
  traversal of configuration automata.
  When the semiring has small domain, we obtain several speedups by exploiting
  advances in matrix-vector multiplication.
  Finally, the speedup of superconfiguration distance extraction on sparse RSMs
  is achieved by devising a Four-Russians type of algorithm, which spends some
  polynomial preprocessing time in order to allow compressing the query input in
  blocks of logarithmic length.
\end{compactitem}

\bigskip\noindent All proofs are provided
in \appendixreftr{Appendix~\ref{app:proofs-sequential}
  and~\ref{app:proofs-concurrent}}.


\section{Preliminaries}
\label{sec:preliminaries}

In this section we present the necessary definitions of recursive state machines
(RSMs) where every transition is labeled with a value (or weight) from an
appropriate domain (semiring).
Then we formally state the problems we study on weighted RSMs.


\myparagraph{Semirings.}
An \emph{idempotent semiring} is a quintuple $\semiring$, where $\Weights$ is a
set called the \emph{domain}, $\zero$ and $\one$ are elements of $\Weights$, and
$\oplus$ (the \emph{combine} operation) and $\otimes$ (the \emph{extend}
operation) are binary operators on $\Weights$ such that
\begin{compactenum}
\item $\tup{\Weights,\oplus,\zero}$ is an idempotent commutative monoid with neutral
  element $\zero$,
\item $\tup{\Weights,\otimes,\one}$ is a monoid with neutral element $\one$,
\item $\otimes$ distributes over $\oplus$,
\item $\zero$ is an annihilator for $\otimes$, i.e.,
  $a \otimes \zero = \zero \otimes a = \zero$ for all $a \in \Weights$.
\end{compactenum}
An idempotent semiring has a canonical partial order $\sleq$, defined by
\begin{align*}
  a \sleq b \iff a \oplus b = a.
\end{align*}
Furthermore, this partial order is \emph{monotonic}, i.e., for all
$a,b,c \in \Weights$
\begin{align*}
  a \sleq b &\implies a \oplus c \sleq b \oplus c, \\
  a \sleq b &\implies a \otimes c \sleq b \otimes c, \\
  a \sleq b &\implies c \otimes a \sleq c \otimes b.
\end{align*}
The \emph{height} $H$ of an idempotent semiring is the length of the longest
descending chain in $\sleq$.
In the rest of the paper we will only write semiring to mean an idempotent
finite-height semiring.

\begin{remark}\label{rem:finite-height_1}
  Instead of finite height, the more general \emph{descending chain
    condition} would be sufficient for our purposes.
  This only requires that there are no infinite descending chains in $\sleq$,
  but there is not necessarily a finite height $\Height$.
\end{remark}


\myparagraph{Recursive State Machines (informally).}
Intuitively, an RSM is a collection of finite automata, called modules, such
that computations consist of ordinary local transitions within a module as well
as calls to other modules, and returns from other modules.
For this, every module has a well-defined interface of entry and exit nodes.
Calls to other modules are represented by boxes, which have call and return
nodes corresponding to the respective entry and exit nodes of the called module.

Unlike pushdown automata (PDAs), there is no explicit stack manipulation in
RSMs.
Instead a call stack is maintained implicitly along computations as follows.
When a call node of a box is reached, the control is passed to the respective
entry node of the called module and the box is pushed onto the top of the stack.
When an exit node of a module is reached, a box is popped off from the top of
the stack and the control is passed to the corresponding return node of the box.
Hence, the stack is a sequence of boxes representing the current calling context
and a configuration in a computation of an RSM is a node together with a
sequence of boxes.

\myparagraph{Recursive State Machines (formally).}
A \emph{recursive state machine (RSM)} over a semiring $\semiring$ is a tuple
$\RSM=\tup{\Module_1,\dots,\Module_k}$, where every \emph{module}
$\Module_i = \tup{B_i,Y_i,\Nodes_i,\delta_i,w_i}$ is given by
\begin{compactitem}
\item a finite set $B_i$ of \emph{boxes},
\item a mapping $Y_i : B_i \Map \set{1,\dots,k}$,
\item a finite set $\Nodes_i = \In_i \cup \En_i \cup \Ex_i \cup \Call_i \cup \Ret_i$ of
  \emph{nodes}, partitioned into
  \begin{compactitem}
  \item \emph{internal} nodes $\In_i$,
  \item \emph{entry} nodes $\En_i$,
  \item \emph{exit} nodes $\Ex_i$,
  \item \emph{call} nodes $\Call_i = \set{\cnd{b}{e} \mid b \in B_i \text{ and } e \in \En_{Y_i(b)}}$,
  \item \emph{return} nodes $\Ret_i  = \set{\rnd{b}{x} \mid b \in B_i \text{ and } x \in \Ex_{Y_i(b)}}$,
  \end{compactitem}
\item a \emph{transition relation}
  $\delta_i \subseteq (\In_i \cup \En_i \cup \Ret_i) \times (\In_i \cup \Ex_i \cup \Call_i)$,
\item a \emph{weight function} $\rsmwt_i : \delta_i \Map \Weights$, with $\rsmwt_i(u,x)=\one$ for every exit node $x\in \Ex_i$.
\end{compactitem}
We write $B$ for $\bigcup_{i=1}^k B_i$, and similarly for $\Nodes$, $\In$,
$\En$, $\Ex$, $\Call$, $\Ret$, $\delta$, $\rsmwt$.
To measure the size of an RSM we let
$|\RSM| = \max(|\Nodes|, \sum_i |\delta_i|)$.
A major source of complexity in analysis problems for RSMs is the number of
entry and exit nodes of the modules.
Throughout the paper we express complexity with respect to the \emph{entry
  bound} $\theta_e = \max_{1 \leq i \leq k} |\En_i|$ and the \emph{exit bound}
$\theta_x = \max_{1 \leq i \leq k} |\Ex_i|$, i.e., the maximum number of entries
and exits, respectively, over all modules.
Note that the restriction on the weight function to assign weight $\one$ to
every transition to an exit node is wlog, as any weighted RSM that does not
respect this can be turned into an equivalent one that does, with only a
constant factor increase in its size.

\myparagraph{Stacks.}
A \emph{stack} is a sequence of boxes $S = b_1 \dots b_r$, where the first box
denotes the top of the stack; and $\varepsilon$ is the empty stack.
The \emph{height} of $S$ is $|S| = r$, i.e, the number of boxes it contains.
For a box $b$ and a stack $S$, we denote with $bS$ the concatenation of $b$ and
$S$, i.e., a push of $b$ onto the top of $S$.

\myparagraph{Configurations and transitions.}  
A \emph{configuration of an RSM $\RSM$} is a tuple $\conf{u}{S}$, where
$u \in \In \cup \En \cup \Ret$ is an internal, entry, or return node, and $S$ is a
stack.
For $S = b_1 \dots b_r$, where $b_i \in B_{j_i}$ for $1 \leq i \leq r$ and some
$j_i$, we require that $Y_{j_i}(b_i) = j_{i-1}$ for $1 < i \leq r$, as well as
$u \in \Nodes_{Y_{j_1}(b_1)}$.
This corresponds to the case where the control is inside the module of node $u$,
which was entered via box $b_1$ from module $\Module_{j_1}$, which was entered
via box $b_2$ from module $\Module_{j_2}$, and so on.
  
We define a transition relation $\Trans$ over configurations and a
corresponding weight function $w : {\Trans} \Map \Weights$ , such that
$\conf{u}{S} \Trans \conf{u'}{S'}$ with
$\rsmwt(\conf{u}{S},\conf{u'}{S'}) = v$ if and only if there exists a
transition $t \in \delta_i$ in $\RSM$ with $\rsmwt_i(t)=v$ and one of the
following holds:
\begin{compactenum}
\item \textit{Internal transition:}
      $u' \in \In_i$,
      $t = \tup{u,u'}$,
      and $S' = S$.
\item \textit{Call transition:}
      $u' = e \in \En_{Y_i(b)}$ for some box $b \in B_i$,
      $t = \tup{u,\cnd{b}{e}}$,
      and $S' = bS$.
\item \textit{Return transition:}
      $u' = \rnd{b}{x} \in R_i$ for some box $b \in B_i$ and exit node $x \in \Ex_{Y_i(b)}$,
      $t = \tup{u,x}$,
      and $S = bS'$.
\end{compactenum}
Note that we follow the convention that a call immediately enters the called
module and a return immediately returns to the calling module.
Hence, the node of a configuration can be an internal node, an entry node, or a
return node, but not a call node or an exit node.

\myparagraph{Computations.}
A \emph{computation} of an RSM $\RSM$ is a sequence of configurations
$\pi = c_1,\dots,c_n$, such that $c_i \Trans c_{i+1}$ for every $1 \leq i < n$.
We say that $\pi$ is a computation from $c_1$ to $c_n$, of length $|\pi| = n-1$,
and of weight $\compwt(\pi) = \bigotimes_{i=1}^{n-1} \rsmwt(c_i,c_{i+1})$ (the
empty extend is $\one$).
We write $\pi : c \Reach{*} c'$ to denote that $\pi$ is a computation from $c$
to $c'$ of any length.
A computation $\pi: c \Reach{*} c'$ is called \emph{non-decreasing} if the stack
height of every configuration of $\pi$ is at least as large as that of $c$ (in
other words, the top stack symbol of $c$ is never popped in $\pi$).
The computation $\pi$ is called \emph{same-context} if it is non-decreasing, and
$c$ and $c'$ have the same stack height.
A computation that cannot be extended by any transition is called a
\emph{halting} computation.
For a set of computations $\Comps$ we define its weight as
$\Compswt(\Comps) = \bigoplus_{\pi \in \Comps} \compwt(\pi)$ (the empty combine
is $\zero$).
For a configuration $c$ and a set of configurations $R$ we denote by
$\Comps(R,c)$ the set of all computations from any configuration in $R$ to $c$.
Here, and for similar purposes below, we will use the convention to write
$\Comps(c,c')$ instead of $\Comps(\set{c},c')$.

\begin{example}
  Figure~\ref{fig:rsm} shows an RSM $\RSM=\tup{\Module_1,\Module_2}$
  that consists of two modules $\Module_1$ and
  $\Module_2$.
  The modules are mutually recursive, since box $b_1$ of module $\Module_1$
  calls module $\Module_2$, and box $b_2$ of module $\Module_2$ calls module
  $\Module_1$.
  A possible computation of $\RSM$ is
  \begin{equation}\label{eq:computation}
  \begin{array}{l}
  \conf{e_1^1}{\eps} \Trans[w_1]
  \conf{e_2}{b_1} \Trans[w_5]
  \conf{e_1^1}{b_2 b_1} \Trans[w_1]
  \conf{e_2}{b_1 b_2 b_1} \Trans[w_6]
  \conf{e_1^2}{b_2 b_1 b_2 b_1} \Trans[w_2]\\
  \conf{u_1}{b_2 b_1 b_2 b_1} \Trans[w_4]
  \conf{\rnd{b_2}{x_1}}{b_1 b_2 b_1} \Trans[w_7]
  \conf{\rnd{b_1}{x_2}}{b_2 b_1} \Trans[w_3]
  \conf{u_1}{b_2 b_1} \Trans[w_4]\\
  \conf{\rnd{b_2}{x_1}}{b_1} \Trans[w_7]
  \conf{\rnd{b_1}{x_2}}{\eps} \Trans[w_3]
  \conf{u_1}{\eps}.
  \end{array}
  \end{equation}
\end{example}

\begin{figure}[tb]
\centering
\input{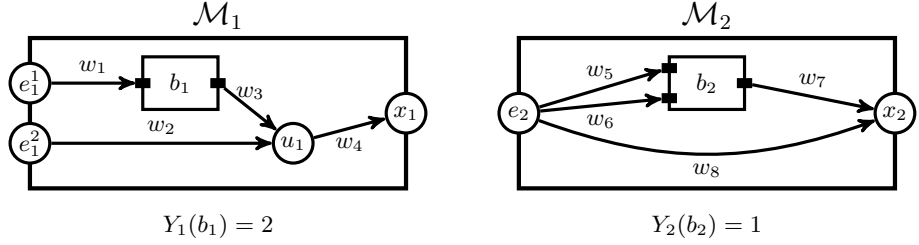}
\begin{tikzpicture}
\def\bend{20}

\node[] at (0, 1.3) {\large $\Module_1$};

\node[csm]  (csm1) at ( 0   ,  0  ) {};
\node[node] (e11)  at (-2.5 ,  0.4) {$e_1^1$};
\node[node] (e12)  at (-2.5 , -0.4) {$e_1^2$};
\node[box]  (b1)   at (-0.5 ,  0.4) {$b_1$};
\node[port] (b1e1) at (-1   ,  0.4) {};
\node[port] (b1x1) at ( 0   ,  0.4) {};
\node[node] (u1)   at ( 1   , -0.4) {$u_1$};
\node[node] (x1)   at ( 2.5 ,  0  ) {$x_1$};

\draw[trans] (e11)  to node [above, midway] {$w_1$} (b1e1);
\draw[trans] (e12)  to node [above, midway] {$w_2$} (u1);
\draw[trans] (b1x1) to node [above, midway] {$w_3$} (u1);
\draw[trans] (u1)   to node [below, midway] {$w_4$} (x1);

\node[] at (0, -1.5) {$Y_1(b_1) = 2$};

\def\ydisp{0}
\def\xdisp{6.5}

\node[] at (0+\xdisp, 1.3+\ydisp) {\large $\Module_2$};

\node[csm]  (csm2) at ( 0   + \xdisp , 0   + \ydisp) {};
\node[node] (e2)   at (-2.5 + \xdisp , 0   + \ydisp) {$e_2$};
\node[box]  (b2)   at ( 0   + \xdisp , 0.4 + \ydisp) {$b_2$};
\node[port] (b2e1) at (-0.5 + \xdisp , 0.6 + \ydisp) {};
\node[port] (b2e2) at (-0.5 + \xdisp , 0.2 + \ydisp) {};
\node[port] (b2x1) at ( 0.5 + \xdisp , 0.4 + \ydisp) {};
\node[node] (x2)   at ( 2.5 + \xdisp , 0   + \ydisp) {$x_2$};

\draw[trans] (e2)   to node [below, midway] {$w_6$} (b2e2);
\draw[trans] (e2)   to node [above, midway] {$w_5$} (b2e1);
\draw[trans] (b2x1) to node [above, midway] {$w_7$} (x2);
\draw[trans, bend right=\bend] (e2) to node [below, midway] {$w_8$} (x2);

\node[] at (0+\xdisp, -1.5+\ydisp) {$Y_2(b_2) = 1$};
\end{tikzpicture}
\caption{Example of a weighted RSM that consists of two modules
  with mutual recursion.}\label{fig:rsm}
\end{figure}



\myparagraph{Distance problems.}
Given a set of configurations $R$, the set of configurations that are
\emph{reachable} from $R$ is
\[
\post^*(R) = \set{ c \mid \exists c_0 \in R : c_0 \Reach{*} c }.
\]
Instead of mere reachability, we are interested in the following distance
metrics that aggregate over computations from $R$ using the semiring combine and
hence are expressed as semiring values.
\begin{compactitem}

\item {\em Configuration distance.}
The \emph{configuration distance} from $R$ to $c$ is defined as
\[
\rsmdist(R,c) = \Compswt(\Comps(R,c)).
\]
That is, we take the combine over the weights of all computations from a
configuration in $R$ to $c$.
Naturally, for configurations $c$ not reachable from $R$ we have
$\rsmdist(R,c) = \zero$.

\item {\em Superconfiguration distance.}
A \emph{superstack} is a sequence of modules
$\ov{S} = \Module_1 \dots \Module_r$.
A stack $S = b_1 \dots b_r$ \emph{refines} $\ov{S}$ if $b_i \in B_i$ for all
$1 \leq i \leq r$, i.e., the $i$-th box of $S$ belongs to the $i$-th module of
$\ov{S}$.
A \emph{superconfiguration} of $\RSM$ is a tuple $\conf{u}{\ov{S}}$.
Let
$\sem{ \conf{u}{\ov{S}} } = \set{ \conf{u}{S} \mid S \text{
    refines } \ov{S} }$.
The \emph{superconfiguration distance} from $R$ to a superconfiguration
$\ov{c}$ is defined as
\[
\rsmdist(R,\ov{c}) = \Compswt_{c \in \sem{ \ov{c}}} \rsmdist(R, c)
\]
The superconfiguration distance is only concerned with the sequence
of modules that have been used to reach the node $u$, rather than the concrete
sequence of boxes as in the configuration distance.
This is a coarser notion than configuration and allows for fast overapproximation.

\item {\em Node and same-context distance.}
The \emph{node distance} of a node $u$ from $R$ is defined as
\[
\rsmdist(R, u) = \Compswt_{c=\conf{u}{S}} \rsmdist(R, c)
\]
where $S$ ranges over stacks of $\RSM$.
Finally, the \emph{same-context node distance} of a node $u$ in module $\Module_i$ is defined as
\[
\rsmdist(\Module_i, u)=\Compswt_{e\in \En_i} \rsmdist(\conf{e}{\eps}, \conf{u}{\eps}).
\]
Intuitively, the node distance minimizes over all possible ways (i.e., stack
sequences) to reach a node, and the same-context problem considers nodes in the
same module that can be reached with empty stack.
\end{compactitem}

\myparagraph{Relevance.}
We discuss the relevance of the model and the problems we consider in program
analysis.
A prime application area of our framework is the analysis of procedural
programs.
Computations in an RSM correspond to the interprocedurally valid paths of a
program.
The distance values defined above allow to obtain information at different
levels of granularity, depending on the requirement for a particular analysis.
MEME (multi-entry, multi-exit) RSMs naturally arise in the model checking of 
procedural programs, where every node represents a combination of control location 
and data.
Checking for reachability, usually of an error state, requires only the simple
Boolean semiring.
On the other hand, interprocedural data flow analysis problems, like in IFDS/IDE,
are usually cast on SESE (single-entry, single-exit) RSMs 
(the control flow graph of the program) using richer semirings.
Our framework captures both of these important applications, and furthermore
allows a hybrid approach of modeling program information both in the state space
of the RSM as well as in the semiring.


\section{Configuration Distance Algorithm}
\label{sec:rsm_sequential}

In this section we present an algorithm which takes as input an RSM $\RSM$ and a
representation of a \emph{regular set} of configurations $R$, and computes a
representation of the set of reachable configurations $\post^*(R)$ that allows
the extraction of the distance metrics defined above.
In Section~\ref{sec:sequential:automata} we introduce configuration automata
as representation structures for regular sets of configurations.
In Section~\ref{sec:sequential:bounded} we present an algorithm for RSMs over
finite-height semirings.
The algorithm saturates the input configuration automaton with additional
transitions and assigns the correct weights via a dynamic programming approach
that gradually relaxes transition weights from an initial overapproximation.
We exploit the monotonicity property in idempotent semirings which allows to
factor the computation into subproblems, and hence corresponds to the
\emph{optimal substructure} property of dynamic programming.
Although a transition might have to be processed multiple times, the finite
height of the semiring prevents a transition from being relaxed indefinitely.
Here we show that the final configuration automata constructed
by our algorithms correctly capture configuration distances.
The extraction of distance values is considered in 
Section~\ref{sec:distance_extraction}.

\subsection{Configuration Automata}
\label{sec:sequential:automata}

In general, like $R$, the set $\post^*(R)$ is infinite.
Hence we make use of a representation of regular sets of configurations as the
language accepted by configuration automata, defined below.
The main feature of a regular set of configurations $R$ is its closure under
$\post^*$. 
That is, $\post^*(R)$ is also a regular set of configurations and can be
represented by a configuration automaton.

\myparagraph{Intuition.}
Every state in a configuration automaton corresponds to a node in the RSM.
In order to represent arbitrary regular sets of configurations we must allow the
replication of states with the same node. 
Therefore we annotate every state with a \emph{mark} (see Remark~\ref{rem:marks}
for details).
Transitions are of two types: (i)~$\eps$-transitions pointing from a node $u$ to
an entry node $e$ and labeled with $\eps$, denoting that a computation reaching
$u$ entered the module of $u$ via entry $e$, and (ii)~b-transitions pointing
from an entry node $e$ to another entry node $e'$ and labeled with a box $b$,
corresponding to a call transition $\tup{u,\cnd{b}{e}}$ in the module of $e'$ in
the RSM.
Reading the labels along a path in the automaton yields a stack.

In addition to the labeling with boxes we label every transition of a
configuration automaton with a semiring value.
In the final configuration automata constructed by our algorithms, every run
generates a configuration $c$ and thereby captures a certain subset
$\Comps \subseteq \Comps(R,c)$ of computations from the initial set of
configurations $R$ to $c$.
The weight of the run equals the combine over the weight of the computations in
$\Comps$.
The combine over the weights of all runs in the automaton that generate $c$
equals the combine over the weights of all computations from $R$ to $c$, i.e.,
the configuration distance $\rsmdist(R, c)$.
Since the transitions in a configuration automaton are essentially reversed
transitions of the RSM (and the extend operation is not commutative), the weight
of a run is given by the extend of the transitions in reversed order.

\myparagraph{Configuration automata.}
Let $\Marks$ be a countably infinite set of \emph{marks}.
A \emph{configuration automaton for an RSM $\RSM$}, also called an
\emph{$\RSM$-automaton}, is a weighted finite automaton
$\Automaton = \tup{Q,B,{\trans},I,F,\autwt}$, where
\begin{compactitem}
\item $Q \subseteq (\In \cup \En \cup \Ret) \times \Marks$ is a finite set
of \emph{states},
\item $B$ (the boxes of $\RSM$) is the transition alphabet,
\item ${\trans} \subseteq Q \times ({B \cup \set{\eps}}) \times Q$ is a
transition relation, such that every transition has one of the following
forms:
\begin{compactitem}
\item \textit{b-transition:}
$\state{e}{m_e} \trans[b] \state{e'}{m_{e'}}$, where
$b \in B_i$ for some $i$, $e \in \En_{Y_i(b)}$, and $e' \in \En_i $,
\item \textit{$\eps$-transition:}
$\state{u}{m_u} \trans[\eps] \state{e}{m_e}$, where
$e \in \En_i$ for some $i$, and either $u \in \In_i \cup \Ret_i$, or $u = e$,
\end{compactitem}
\item $I \subseteq Q$ is a set of \emph{initial states},
\item $F \subseteq Q$ and $F \subseteq \En \times \Marks$ is a set of \emph{final states},
\item $\autwt : {\trans} \Map \Weights$ is a \emph{weight function} that assigns a
semiring weight to every transition.
\end{compactitem}

\begin{remark}[Marks]\label{rem:marks}
  The marks in the states of a configuration automaton are introduced to support
  the general setting of representing an arbitrary set of configurations, e.g.,
  with stacks that are not even reachable in the RSM.
  Since every state is tied to an RSM node, the marks allow to have multiple
  ``copies'' of the same node in unrelated parts of the automaton.
  Furthermore, our algorithm (Section~\ref{sec:sequential:bounded}) introduces a
  \emph{fresh mark} to recognize when it can safely store \emph{entry-to-exit
    summaries}.
  For the common setting of starting the analysis from the entry nodes of a main
  module with empty stack, marks are not necessary and can be elided.
\end{remark}

\myparagraph{Runs and regular sets of configurations.}
A \emph{run} of a configuration automaton $\Automaton$ is a sequence
$\lambda = t_1,\dots,t_{n}$, such that there are states $q_1,\dots,q_{n+1}$ and each
$t_i = q_i \trans[\sigma_i] q_{i+1}$ is a transition of $\Automaton$ labeled
with $\sigma_i$.
We say that $\lambda$ is a run from $q_1$ to $q_{n+1}$, of length $|\lambda| = n$,
labeled by $S = \sigma_1 \dots \sigma_n$, and of weight
$\runwt(\lambda) = \bigotimes_{i=n}^{1} \autwt(t_i)$
%
(note that the weights of the transitions are extended in reverse order).
We write $\lambda : q \reach[S / v]{*} q'$ to denote that $\lambda$ is a run
from $q$ to $q'$ of any length labeled by $S$ and of weight $v$.
We will also use the notation without $v$ if we are not interested in the
weight.
The run $\lambda$ is \emph{accepting} if $q \in I$ and $q' \in F$.
A configuration $\conf{u}{S}$ is \emph{accepted} by $\Automaton$ if there is an
accepting run $\lambda: \state{u}{m} \reach[S]{*} q_f$ for some mark
$m \in \Marks$, and additionally $\runwt(\lambda) \neq \zero$.
We say that two runs are \emph{equivalent} if they accept the same configuration
with the same weight.
For technical convenience we consider that for every state $\state{e}{m_e}$ with
entry node $e \in \En$ there is an $\eps$-self-loop
$\state{e}{m_e} \trans[\eps] \state{e}{m_e}$ with weight $\one$.

The set of all configurations accepted by $\Automaton$ is denoted by
$\lang(\Automaton)$.
A set of configurations $R$ is called \emph{regular} if there exists an
$\RSM$-automaton $\Automaton$ such that $\lang(\Automaton)= R$.
For a configuration $c$ let $\Runs(c)$ be the set of all accepting runs of $c$
and define $\Automaton(c) = \bigoplus_{\lambda \in \Runs(c)} \runwt(\lambda)$
the weight that $\Automaton$ assigns to $c$.

We note that, despite the imposed syntactic restrictions, our definition of
configuration automata is most general in the following sense.

\begin{proposition}
Let $R$ be a set of configurations such that their string representations
is a regular language.
Then there exists a configuration automaton $\Automaton$ such that
$\lang(\Automaton) = R$.
\end{proposition}

\subsection{Algorithm for Finite-Height Semirings}
\label{sec:sequential:bounded}

In the following we present algorithm $\ConfDistBounded$ for computing the set
$\post^*(R)$ of a regular set of configurations $R$.
The algorithm operates on an $\RSM$-automaton $\Automaton$ with
$\lang(\Automaton)=R$.
In the end, it has constructed an $\RSM$-automaton $\Automaton_{\post^*}$ such
that $\lang(\Automaton_{\post^*})=\post^*(R)$.
Moreover, the configuration distance $\rsmdist(R,c)$ from $R$ to any
configuration $c$ can be obtained from the labels of $\Automaton_{\post^*}$ as
$\Automaton_{\post^*}(c)$. 
A computation is called \emph{initialized}, if its first configuration is
accepted by the initial configuration automation $\Automaton$.

\myparagraph{Key technical contribution.}
In this work we consider the configuration distance computation.
Using the notion of configuration automata as a \emph{symbolic} representation
structure for regular sets of configurations, the solution of the configuration
distance problem has been previously studied in the setting of (weighted)
pushdown systems~\cite{Schwoon02,Reps05,Bouajjani05}.
One of the main algorithmic ideas for the efficient RSM reachability algorithm
of~\cite{AlurBEGRY05} is to expand RSM transitions and use entry-to-exit
\emph{summaries} to avoid traversing a module more than once.
However, the algorithm in~\cite{AlurBEGRY05} is limited to the node reachability
problem. 
We combine the symbolic representation of configuration automata, along with the
summarization principle, to obtain an efficient algorithm for the
general configuration distance problem on RSMs.

\myparagraph{Intuitive description of $\ConfDistBounded$.}
The intuition behind our algorithm is very simple: it performs a forward search
in the RSM. 
In every iteration it picks a frontier node $u$ and extends the already
discovered computations to $u$ with the outgoing transitions from $u$. 
Depending on the type of outgoing transitions, a new node discovered and added
to the frontier can be (a)~an internal node by following an internal transition,
(b)~the entry node of another module by following a call transition, and (c)~a
return node corresponding to a previously discovered call by following an exit
transition.

In summary, the algorithm simply follows interprocedural paths. 
However, the crux to achieve our complexity is to keep summaries of paths
through a module. 
Whenever we discovered a full (interprocedural) path from an entry $e$ to an
exit $x$, we keep its weight as an upper bound.
Now any subsequently discovered call reaching $e$ does not need to continue the
search from $e$, but short-circuits to $x$ by using the stored summary.

\myparagraph{Preprocessing.}
In order to ease the formal presentation of the algorithm, we consider the
following preprocessing on the initial configuration automaton $\Automaton$.
Let $M \subseteq \Marks$ be the set of marks in the initial automaton and
$\freshMark \in \Marks \setminus M$ a \emph{fresh mark}.
\begin{compactenum}
\item For every node $u \in \In \cup \En \cup \Ret$, we add a new state
  $\state{u}{\freshMark}$ marked with the fresh mark.
  Additionally, all these new states are declared initial.
\item For every initial state $\state{u}{m_u} \in I$ such that there is a call
  transition $t = \tup{u,\cnd{b}{e}} \in \delta_i$ in $\RSM$, for every state
  $\state{e'}{m_{e'}}$ where $e'$ is an entry node of the same module as $u$, we
  add a b-transition $\state{e}{\freshMark} \trans[b] \state{e}{m_{e'}}$ with
  weight $\zero$.
\item For every state $\state{e}{m_e}$ with entry node $e \in \En_i$ and every
  internal or return node $u \in \In_i \cup \Ret_i$ in the same module as $e$,
  we add an $\eps$-transition
  $\state{u}{\freshMark} \trans[\eps] \state{e}{m_e}$ with weight $\zero$.
\end{compactenum}
Essentially the preprocessing a priori adds to $\Automaton$ all possible states
and transitions, so that the algorithm only has to relax those transitions (i.e.,
without adding them first).
Note that the preprocessing only provides for an easier presentation of our
algorithm.
Indeed, in practice it would be impractical to do the full preprocessing and
thus our implementation adds states and transitions to the automaton on the fly.

\myparagraph{Technical description of $\ConfDistBounded$.}
We present a detailed explanation of the algorithm supporting the formal
description given in Algorithm~\ref{alg:rsm_post}.
We require that every transition in the input configuration automaton
$\Automaton$ has weight $\one$, since the configurations in $\lang(\Automaton)$
should not contribute any initial weight to the configuration distance.
The algorithm maintains a \emph{worklist} $\WL$ of weighted transitions either
of the form $\state{u}{m_u} \trans[\eps] \state{e}{m_e}$ or
$\state{e}{m_e} \trans[b] \state{e'}{m_{e'}}$, and a \emph{summary function}
$\summary : (\En \times \Marks) \times \Ex \Map \Weights$.
Initially, the worklist contains all such transitions where the source state
$\state{u}{m_u}$ is an initial state in $I$, and $\summary$ is all $\zero$.
In every iteration a transition $t_\Automaton$ is extracted from the
worklist and processed as follows.
Since every accepting run starting with $t_\Automaton$ corresponds to a
reachable configuration $\conf{u}{S}$ (where $S$ varies over different runs),
every transition $t_{\RSM}=\tup{u,u'}$ in $\RSM$ gives rise to another reachable
configuration.
More precisely, the run corresponds to a set of computations reaching
$\conf{u}{S}$ from the initial set of configurations, and $t_\RSM$ allows to
extend these computations by one step.
The algorithm incorporates the newly discovered computations by relaxing a
transition as follows, illustrated in Figure~\ref{fig:relax}.

\tikzstyle{relaxpicture}=[->,>=stealth',auto]
\tikzstyle{reltrans}=[thick]
\tikzstyle{exttrans}=[]
\tikzstyle{fintrans}=[dotted]
\tikzstyle{mwa}=[midway,above]
\tikzstyle{nsa}=[near start,above]

\newcommand{\qflabel}{}
\newcommand{\leftWidth}{0.53}
\newcommand{\rightWidth}{0.46}

\begin{figure}[tb]
  \small

  \begin{minipage}{\leftWidth\linewidth}
    \centering
    \fbox{$t_\RSM = \tup{u,u'}$}
    
    \begin{tikzpicture}[relaxpicture]
      \node [] (u)  at (0.0 , 1.0) {$\state{u}{m_u}$};
      \node [] (e)  at (1.8 , 1.0) {$\state{e}{m_e}$};
      \node [] (qf) at (3.1 , 1.0) {\qflabel};
      \node [] (up) at (0.5 , 0.2) {$\state{u'}{\freshMark}$};
      
      \path
      (u)  edge [exttrans] node [mwa] {$\eps$} (e)
      (e)  edge [fintrans] node [mwa] {$S$}    (qf)
      (up) edge [reltrans] node [nsa] {$\eps$} (e);
    \end{tikzpicture}
    
    (a) Internal transition (Lines~\ref{alg:rsm_post:internal_trans_loop}--\ref{alg:rsm_post:relax1})
  \end{minipage}
  \begin{minipage}{\rightWidth\linewidth}
    \centering
    \fbox{$t_\RSM = \tup{u,\cnd{b}{e'}}$}

    \begin{tikzpicture}[relaxpicture]
      \node [] (u)  at (0.0 , 1.0) {$\state{u}{m_u}$};
      \node [] (e)  at (1.8 , 1.0) {$\state{e}{m_e}$};
      \node [] (qf) at (3.1 , 1.0) {\qflabel};
      \node [] (ep) at (0.5 , 0.2) {$\state{e'}{\freshMark}$};
      
      \path
      (u)  edge [exttrans] node [mwa] {$\eps$} (e)
      (e)  edge [fintrans] node [mwa] {$S$}    (qf)
      (ep) edge [reltrans] node [nsa] {$b$} (e);
    \end{tikzpicture}
    
    (b) Call transition (Lines~\ref{alg:rsm_post:call_trans_loop}--\ref{alg:rsm_post:add_eps_loop})
  \end{minipage}

  \vspace{.5cm}

  \begin{minipage}{\leftWidth\linewidth}
    \centering
    \fbox{$t_\RSM = \tup{u,x}$}
    
    \begin{tikzpicture}[relaxpicture]
      \node [] (u)  at (0.0 , 1.0) {$\state{u}{m_u}$};
      \node [] (e)  at (1.7 , 1.0) {$\state{e}{m_e}$};
      \node [] (e1) at (3.5 , 1.0) {$\state{e'}{m_{e'}}$};
      \node [] (qf) at (4.8 , 1.0) {\qflabel};
      \node [] (r)  at (2.0 , 0.2  ) {$\state{\rnd{b}{x}}{\freshMark}$};
      
      \path
      (u)  edge [exttrans] node [mwa] {$\eps$} (e)
      (e)  edge []         node [mwa] {$b$}    (e1)
      (e1) edge [fintrans] node [mwa] {$S'$}   (qf)
      (r)  edge [reltrans] node [nsa] {$\eps$} (e1);
    \end{tikzpicture}
    
    (c) Exit transition (Lines~\ref{alg:rsm_post:exit_trans_loop}--\ref{alg:rsm_post:relax3})
  \end{minipage}
  \begin{minipage}{\rightWidth\linewidth}
    \centering
    \fbox{$\summary(\state{e}{\freshMark},x)$}

    \begin{tikzpicture}[relaxpicture]
      \node [] (e)  at (0.0 , 1.0) {$\state{e}{\freshMark}$};
      \node [] (ep) at (1.8 , 1.0) {$\state{e'}{m_{e'}}$};
      \node [] (qf) at (3.1 , 1.0) {\qflabel};
      \node [] (r)  at (0.5 , 0.2) {$\state{\rnd{b}{x}}{\freshMark}$};
      
      \path
      (e)  edge [exttrans] node [mwa] {$b$} (ep)
      (ep)  edge [fintrans] node [mwa] {$S$}    (qf)
      (r) edge [reltrans] node [nsa] {$\eps$} (ep);
    \end{tikzpicture}
    
    (d) Using summary (Lines~\ref{alg:rsm_post:process_b}--\ref{alg:rsm_post:relax4})
  \end{minipage}
  
  \caption{Relaxation steps of $\ConfDistBounded$.}
  \label{fig:relax}
\end{figure}
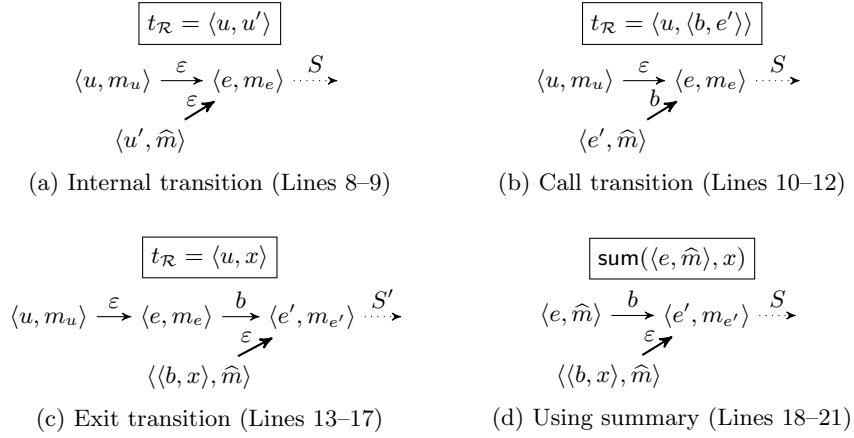


\begin{enumerate}[parsep=0pt]
\item If $t_{\Automaton}$ is of the form
  $\state{u}{m_u} \trans[\eps] \state{e}{m_e}$, then:

  \begin{enumerate}[series=relax:cases]
  \item If $u'$ is an internal node then the algorithm captures the internal
    transition $\conf{u}{S} \Trans \conf{u'}{S}$ by relaxing the transition
    $\state{u'}{\freshMark} \trans[\eps] \state{e}{m_e}$ using the weights
    $\autwt(t_\Automaton)$ and $\rsmwt(t_{\RSM})$.

  \item If $u'$ is a call node $\cnd{b}{e'}$ then the transition
    $\state{e'}{\freshMark} \trans[b] \state{e}{m_e}$ is relaxed with the new weight
    $\autwt(t_\Automaton) \otimes \rsmwt(t_\RSM)$.
    Furthermore, an $\eps$-self-loop is stored in the worklist to continue
    exploration from the called entry node $e'$.
    
  \item If $u'$ is an exit node $x$ then the algorithm relaxes
    $\summary(\state{e}{m_e},x)$ if a smaller computation to $x$ has been
    discovered.
    Note that for $m_e = \freshMark$ this corresponds to valid entry-to-exit
    computations from $e$ to $x$.
    If another call to $e$ is discovered later, the summary is used to avoid
    traversing the module again.
    For $m_e \neq \freshMark$ the summary does not necessarily correspond to
    valid entry-to-exit computations (e.g., because node $u$ was provided as an
    initial configuration) and is only stored to avoid redundant work.
    
    For a return transition from $\conf{u}{S}$ the stack $S$ has to be
    non-empty.
    The algorithm looks for all possible boxes $b$ at the top of $S$ by going
    along a $b$-transition from $\state{e}{m_e}$ to a state
    $\state{e'}{m_{e'}}$.
    Then for any $S = bS'$, relaxing the transition
    $\state{\rnd{b}{x}}{\freshMark} \trans[\eps] \state{e'}{m_{e'}}$ captures
    the return transition $\conf{u}{S} \Trans \conf{\rnd{b}{x}}{S'}$.
    Note that here we make use of the fact that the return transition itself has
    weight $\one$.
  \end{enumerate}
  
\item If $t_{\Automaton}$ is of the form
  $\state{e}{m_e} \trans[b] \state{e'}{m_{e'}}$, then:
  
  \begin{enumerate}[resume*=relax:cases]
  \item for every exit node $x$ in the module of $e$ the summary function
    is used to relax the weight of the transition
    $\state{\rnd{b}{x}}{\freshMark} \trans[\eps] \state{e'}{m_{e'}}$ to the
    value $\autwt(t_\Automaton) \otimes \summary(\state{e}{\freshMark},x)$.
  \end{enumerate}
\end{enumerate}

The initial states of $\Automaton_{\post^*}$ are the initial states of
$\Automaton$ together with all states with the fresh mark added in the
preprocessing.
The final states of $\Automaton_{\post^*}$ are the unmodified final states of
$\Automaton$.

\newcommand\mycommfont[1]{\ttfamily\textcolor{blue}{#1}}
\SetCommentSty{mycommfont}

\begin{algorithm}[!h]\small
  \caption{$\ConfDistBounded$}
  \label{alg:rsm_post}
  \KwIn{RSM $\RSM$ and $\RSM$-automaton $\Automaton$ with $\autwt(t) = \one$ for all transitions $t$ in $\Automaton$}
  \KwOut{$\RSM$-automaton $\Automaton_{\post^*}$ with $\Automaton_{\post^*}(c) = \rsmdist(\lang(\Automaton), c)$ for all configurations $c$} 

  \SetKwProg{RelaxProc}{Procedure}{}{}

  \BlankLine

  preprocess $\Automaton$ as described in the main text\;
  \tcp{Initialization of worklist and summary function}
  $\WL \ass$ $\set{t = q \trans[\eps] q' \mid q \in I \text{ and } \autwt(t) = \one}$\;\label{alg:rsm_post:init_wl}
  $\summary(\state{e}{m_e},x) \ass \zero$ for all states $\state{e}{m_e}$ and $x \in \Ex$

  \tcp{Main loop}
  \While{$\WL \neq \emptyset$}{\label{alg:rsm_post:while}
    extract $t_\Automaton$  from $\WL$\;\label{alg:rsm_post:extract_wl}

    \uIf{$t_\Automaton = \state{u}{m_u} \trans[\eps] \state{e}{m_e}$}{\label{alg:rsm_post:process_eps}
      let $\Module_i$ be the module of node $u$\;
      \tcp{Internal transitions from $u$}
      \ForEach{$t_\RSM = \tup{u,u'} \in \delta_i$ where $u' \in \In_i$}{\label{alg:rsm_post:internal_trans_loop}
        $\Relax(\state{u'}{\freshMark} \trans[\eps] \state{e}{m_e},
        \autwt(t_\Automaton) \otimes \rsmwt_i(t_\RSM))$\;\label{alg:rsm_post:relax1}
      }

      \tcp{Call transitions from $u$}
      \ForEach{$t_\RSM = \tup{u,\cnd{b}{e'}} \in \delta_i$}{\label{alg:rsm_post:call_trans_loop}
        $\Relax(\state{e'}{\freshMark} \trans[b] \state{e}{m_e},
        \autwt(t_\Automaton) \otimes \rsmwt_i(t_\RSM))$\;\label{alg:rsm_post:relax2}
        add $\state{e'}{\freshMark} \trans[\eps] \state{e'}{\freshMark}$ to $\WL$, if it was never added before\label{alg:rsm_post:add_eps_loop}
      }

      \tcp{Exit transitions from $u$}
      \ForEach{$t_\RSM = \tup{u,x} \in \delta_i$ where $x \in \Ex_i$}{\label{alg:rsm_post:exit_trans_loop}
        \If{$\summary(\state{e}{m_e}, x) \not \sleq \autwt(t_\Automaton) $}{\label{alg:rsm_post:summary_if}
          $\summary(\state{e}{m_e},x) \ass \summary(\state{e}{m_e},x)
          \oplus \autwt(t_\Automaton) $\;\label{alg:rsm_post:summarize}
          \ForEach{$\state{e}{m_e} \trans[b/v] \state{e'}{m_{e'}}$}{\label{alg:rsm_post:entry_exit_trans_loop}
            $\Relax(\state{\rnd{b}{x}}{\freshMark} \trans[\eps] \state{e'}{m_{e'}}, v \otimes \summary(\state{e}{m_e},x))$\label{alg:rsm_post:relax3}
          }
        }
      }
    }
    \uElseIf{$t_\Automaton = \state{e}{m_e} \trans[b] \state{e'}{m_{e'}}$}{\label{alg:rsm_post:process_b}
      let $\Module_i$ be the module of node $e$\;
      \tcp{Using entry-to-exit summaries}
      \ForEach{$x \in \Ex_i$}{\label{alg:rsm_post:call_return_loop}
        $\Relax(\state{\rnd{b}{x}}{\freshMark} \trans[\eps] \state{e'}{m_{e'}},
        \autwt(t_\Automaton) \otimes \summary(\state{e}{\freshMark},x)$\;\label{alg:rsm_post:relax4}
      }
    }
  }

  \BlankLine

  \RelaxProc{$\Relax(t, v)$}{
    \If{$\autwt(t) \neq \autwt(t) \oplus v$}{
      $\autwt(t) \ass \autwt(t) \oplus v$\;
      add $t$ to $\WL$\label{alg:rsm_post:add_wl}
    }
  }
\end{algorithm}


\vspace{-.5cm}

\begin{example}
In Figure~\ref{fig:confaut_reachability} we illustrate an execution of
$\ConfDistBounded$ for the reachability problem in the RSM from
Figure~\ref{fig:rsm}.
The reader can verify that every configuration in the example
computation~\eqref{eq:computation} is accepted by a run of the constructed
automaton.
\end{example}

\begin{figure}[t]
\centering
\small
\input{graphics/confaut_style}
\begin{minipage}[c]{.41\textwidth}
  \caption{The configuration automaton $\Automaton_{\post^*}$ constructed by
    $\ConfDistBounded$ for the RSM in Figure~\ref{fig:rsm} over the Boolean
    semiring $\booleanSemiring$, expressing the reachability problem.
    The initial input automaton $\Automaton$ is given by the black states,
    whereas the gray states represent the newly added states with the fresh mark
    $\freshMark$.
    The black/gray color gives a similar distinction for the transitions (i.e.,
    the gray transitions have been added by the algorithm).
    The set of initial states of $\Automaton$ is $I=\{e_1^1, e_2\}$, and the set
    of final states is the singleton set $F=\{e_1^1\}$.
    Transitions added in the preprocessing phase with value $\zero$
    are not shown.}\label{fig:confaut_reachability}
\end{minipage}
\begin{minipage}[c]{.58\textwidth}
\begin{tikzpicture}[scale=.7]
\def\bend{10}
\def\ystep{2.8}
\def\xstep{3}

\node[state_init]                 (e2m)    at (0,0)                {$e_2$};
\node[state_init]                 (e12m)   at (0,-1*\ystep)        {};
\node[state_init,minimum size=21] ()       at (0,-1*\ystep)        {$e_1^1$};
\node[state_added]                (e11mh)  at (1*\xstep,-0*\ystep) {$e_1^1$};
\node[state_added]                (b1x2mh) at (1*\xstep,-1*\ystep) {$\rnd{b_1}{x_2}$};
\node[state_added]                (e2mh)   at (2*\xstep,-0*\ystep) {$e_2$};

\node[state_added] (e12mh) at (1*\xstep,1*\ystep) {$e_1^2$};
\node[state_added] (u1mh)  at (2*\xstep,-1*\ystep) {$u_1$};

\node[state_added] (b2x1mh) at (0*\xstep,1*\ystep) {$\rnd{b_2}{x_1}$};

\draw[trans_init]  (e2m)    to node[left,  midway]{$b_1$}  (e12m);
\draw[trans_added] (b2x1mh) to node[left,  midway]{$\eps$} (e2m);
\draw[trans_added] (b1x2mh) to node[above, midway]{$\eps$} (e12m);
\draw[trans_added] (b1x2mh) to node[right, midway]{$\eps$} (e11mh);

\draw[trans_added] (u1mh) to node[above, near start]{$\eps$} (e11mh);
\draw [trans_added] (u1mh) to ([xshift=1.2cm]u1mh) to node[right, pos=0.6]{$\eps$}  ([xshift=1.2cm]2*\xstep,1*\ystep) to (e12mh);
\draw[trans_added, , bend left=28] (u1mh) to node[below, midway]{$\eps$} (e12m);

\draw[trans_added, bend right=\bend] (e2mh)  to node[above, midway]{$b_1$} (e11mh);
\draw[trans_added, bend right=\bend] (e11mh) to node[below, midway]{$b_2$} (e2mh);

\draw[trans_added] (e12mh) to node[above, midway]{$b_2$} (e2mh);
\draw[trans_added] (e11mh) to node[below, midway]{$b_2$} (e2m);
\draw[trans_added] (e12mh) to node[above, midway]{$b_2$} (e2m);

\draw [trans_added] (b2x1mh) to ([yshift=1.5cm]b2x1mh) to node[above, pos=0.6]{$\eps$} ([yshift=1.5cm, xshift=0.8cm]2*\xstep,1*\ystep) to ([xshift=0.8cm]e2mh) to (e2mh);

\draw [trans_added, out=160, in=200, looseness=5] (e2m)   to node[left,  midway]{$\eps$} (e2m);
\draw [trans_added, out=70, in=110, looseness=5]  (e12mh) to node[above, midway]{$\eps$} (e12mh);
\draw [trans_added, out=70, in=110, looseness=5]  (e11mh) to node[above, midway]{$\eps$} (e11mh);
\draw [trans_added, out=250, in=290, looseness=5] (e2mh)  to node[below, midway]{$\eps$} (e2mh);

\end{tikzpicture}
\end{minipage}
\end{figure}


\myparagraph{Correctness.}
In the following we outline the correctness of the algorithm.
We start with a simple observation about the shape of runs in the constructed
configuration automaton.

\begin{proposition}\label{prop:run-shape}
  For every accepting run $\lambda$ there exists an equivalent accepting run
  $\lambda'$ that starts with an $\eps$-transition followed by only
  b-transitions. 
  Furthermore, all but the first state contain an entry node.
\end{proposition}

The following three lemmas capture the correctness of $\ConfDistBounded$.
We start with completeness, namely that the distance computed for any
configuration $c$ is at most the actual distance from the initial set of
configurations $\lang(\Automaton)$ to $c$.
The proof relies on showing that for any initialized computation
$\pi : \conf{u}{S} \Reach{*} \conf{u'}{S'}$ there is a run $\lambda$ accepting
$\conf{u'}{S'}$ such that $\runwt(\lambda) \sleq \compwt(\pi)$, and follows an
induction on the length $|\pi|$.

\begin{restatable}[Completeness]{lemma}{lemcompleteness}\label{lem:completeness}
  For every configuration $c$ we have
  $\Automaton_{\post^*}(c) \sleq \rsmdist(\lang(\Automaton), c)$.
\end{restatable}


We now turn our attention to soundness, namely that the distance computed for
any configuration $c$ is at least the actual distance from the initial set of
configurations $\lang(\Automaton)$ to $c$.
The proof is established via a set of interdependent invariants that state that
the algorithm maintains sound entry-to-exit summaries and any run in the
automaton has a weight that is witnessed by a set of computations.

\begin{restatable}[Soundness]{lemma}{lemsoundness}\label{lem:soundness}
  For every configuration $c$ we have
  $\rsmdist(\lang(\Automaton), c) \sleq \Automaton_{\post^*}(c)$.
\end{restatable}


\myparagraph{Complexity.}
Finally, we turn our attention to the complexity analysis of the algorithm,
which is done by bounding the number of times the algorithm can perform a
relaxation step.
The complexity bound is based on the height of the semiring $\Height$, which implies
that every transition can be relaxed at most $\Height$ times.
The contribution of the size of the initial automaton $\Automaton$ in the
complexity is captured by the number of initial marks $\kappa$.

\begin{restatable}[Complexity]{lemma}{lemcomplexity}\label{lem:complexity}
  Let $\kappa$ be the number of distinct marks $m\in \Marks$ of the initial
  automaton $\Automaton$.
  Algorithm $\ConfDistBounded$ constructs $\Automaton_{\post^*}$ in time
  $O(\Height\cdot (|\RSM|\cdot \theta_e\cdot \kappa^2 + |\Call|\cdot \theta_e\cdot \theta_x\cdot \kappa^3))$, and
  $\Automaton_{\post^*}$ has $O(|\RSM|\cdot \theta_e\cdot \kappa^2)$  transitions.
\end{restatable}


\noindent
We summarize the results of this section in the following theorem.

\begin{theorem}\label{them:finite-height}
Let $\RSM$ be an RSM over a semiring of height $\Height$, and $\Automaton$ an
$\RSM$-automaton with $\kappa$ marks.
Algorithm $\ConfDistBounded$ constructs in
$O(\Height \cdot (|\RSM|\cdot \theta_e\cdot \kappa^2 + |\Call|\cdot \theta_e\cdot \theta_x\cdot \kappa^3))$ time an
$\RSM$-automaton $\Automaton_{\post^*}$ with $\kappa+1$ marks, such that
$\rsmdist(\lang(\Automaton),c) = \Automaton_{\post^*}(c)$ for every
configuration $c$.
\end{theorem}

\begin{remark}[Comparison with existing work]\label{rem:finite-height}
We now relate Theorem~\ref{them:finite-height} with the existing work for
computing configuration distance (often called generalized reachability in the
literature) in weighted pushdown systems \emph{(WPDS)}~\cite{Schwoon02,Reps05}.
For simplicity we assume that the initial automaton is of constant size.
A formal description of WPDS is omitted; the reader can refer
to~\cite{AlurBEGRY05,Reps05}.
Let $\WPDS$ be a WPDS where:
\begin{compactenum}
\item $n_{\WPDS}$ is the number of states
\item $n_{\Delta}$ is the size of the transition relation
\item $n_\swappush$ is the number of different pairs $\tup{p', \gamma'}$
such that there is a transition of the form
$\tup{p,\gamma}\trans{} \tup{p', \gamma' \gamma''}$ (i.e., from some
state $p$ with $\gamma$ on the top of the stack, the WPDS $\WPDS$
(i)~transitions to state $p'$, (ii)~swaps $\gamma$ and $\gamma''$, and
(iii)~pushes $\gamma'$ on the top of the stack).
\end{compactenum}
As shown in~\cite{Reps05}, given a WPDS $\WPDS$ with weights from a semiring
with height $\Height$, together with a corresponding automaton
$\Automaton^{\WPDS}$ that encodes configurations of $\WPDS$, an automaton
$\Automaton^{\WPDS}_{\post^*}$ can be constructed as a solution to the
configuration distance problem for $\WPDS$.
For ease of presentation we focus on the common case where $\Automaton^{\WPDS}$ has constant size 
(e.g., for encoding an initial configuration of $\WPDS$ with empty stack).
Then the time required to construct $\Automaton^{\WPDS}_{\post^*}$ is
$O(\Height\cdot n_{\WPDS}\cdot n_{\Delta} \cdot n_\swappush)$~\cite{Schwoon02,Reps05}.

A direct consequence of~\cite[Theorem~1]{AlurBEGRY05} is that
an RSM $\RSM$ and a configuration automaton $\Automaton^{\RSM}$ can be converted to an equivalent PDS $\WPDS$ and configuration automaton $\Automaton^{\WPDS}$, and vice versa, such that the following equalities hold:
\[
|\RSM|=\Theta(n_{\Delta});\quad  \theta_x = \Theta(n_{\WPDS}); \quad f\cdot \theta_e= \Theta(n_\swappush),
\]
where $f$ represents the number of modules.
Hence, the bound we obtain by translating the input RSM to a WPDS and using the algorithm of~\cite{Schwoon02,Reps05}
is $O(H\cdot |\RSM|\cdot \theta_e\cdot \theta_x\cdot f)$.
Our complexity bound on Theorem~\ref{them:finite-height} is better by a factor $\Omega((|\RSM|\cdot f)/(\theta_x + |\Call|))$.
Moreover, to verify such improvements, we have also constructed a family of dense RSMs, and apply 
our algorithm, and compare against the jMoped implementation of the existing algorithms, and 
observe a linear speed-up (see Section~\ref{subsec:example} for details).

The above analysis considers an explicit model, where $\RSM$ comprises two
parts, a program control-flow graph $\RSM_\mathrm{CFG}$ and the set of all data
valuations $V$, where $|V|=\theta_e=\theta_x$.
Hence, $|\RSM|=|\RSM_\mathrm{CFG}|\cdot |V|^2$.
In a symbolic model, where all the data valuations are tracked on the semiring,
the input RSM is a factor $|V|^2$ smaller (i.e., the contribution of the data
valuation to $|\RSM|$), and $\theta_e=\theta_x=1$.
However, now each semiring operation incurs a factor $|V|^2$ increase in time
cost, and the height of the semiring increases by a factor $|V|^2$ as well, in
the worst case.
Hence, existing symbolic approaches for PDSs have the same worst-case time
complexity as the explicit one, and our comparison applies to these as well.
For further discussion on symbolic extensions of our algorithm we refer
to \appendixreftr{Appendix~\ref{app:discussion-symbolic}}.
\end{remark}


\section{Distance Extraction}\label{sec:distance_extraction}

The algorithm presented in Section~\ref{sec:rsm_sequential} takes as input a
weighted RSM $\RSM$ over a semiring and a configuration automaton $\Automaton$
that represents a regular set $R$ of configurations of $\RSM$, and outputs an
automaton $\Automaton_{\post^*}$ that encodes the distance $\rsmdist(R, c)$ to
every configuration $c$.
We now discuss the algorithmic problem of extracting such distances from
$\Automaton_{\post^*}$, and present fast algorithms for this problem.
First we will consider the general case for RSMs over an arbitrary semiring. 
Then we present several improvements for special cases, like RSMs over a
semiring with small domain, or sparse RSMs.
As the correctness of the constructions is straightforward, our attention will
be on the complexity.

\subsection{Distances over General Semirings}\label{subsec:semiring_dist}

\myparagraph{Configuration distances.}
Given a configuration $c=\conf{u}{S}$, $S = b_1 \dots b_{|S|}$, the task is to
extract
$
\rsmdist(R,c) = \Compswt(\Comps(R,c)).
$
This is done by a dynamic-programming style algorithm, which computes
iteratively for every prefix $b_1 \dots b_i$ of $S$ and state $\state{e}{m_e}$
with $e \in \En_j$ and $b_i \in B_j$, the weight
\[
w_{\state{e}{m_e}} = \Compswt \set{ \runwt(\lambda) \mid \lambda:\state{u}{m_u} \reach[b_1\dots b_i]{*}  \state{e}{m_e} }.
\]
Since there are $O(\kappa^2\cdot \theta_e^2)$ transitions labeled with $b_i$,
every iteration requires $O(\kappa^2\cdot \theta_e^2)$ time, and the total time
for computing $\rsmdist(R, c)$ is $O(|S|\cdot \kappa^2\cdot \theta_e^2)$.

\myparagraph{Superconfiguration distances.}
Given a superconfiguration $\ov{c}=\conf{u}{\ov{S}}$,
$\ov{S} = \Module_1 \dots \Module_{|\ov{S}|}$, the task is to extract
$
\rsmdist(R,\ov{c}) = \Compswt_{c \in \sem{\conf{u}{\ov{S}}}} \rsmdist(R, c).
$
To handle such queries, we perform a one-time preprocessing of
$\Automaton_{\post^*}$, so that the transitions are labeled with modules instead
of boxes.
That is, we create an automaton $\ov{\Automaton}_{\post^*}$, initially
identical to $\Automaton_{\post^*}$.
Then we add a transition
$t=\state{e}{m_{e}} \trans[\Module] \state{e'}{m_{e'}}$, with $\Module$ being
the module of $e'$, if there exists a b-transition
$\state{e}{m_{e}} \trans[b] \state{e'}{m_{e'}}$ in $\Automaton_{\post^*}$.
The weight function $\ov{\autwt}$ of $\ov{\Automaton}_{\post^*}$ is such that
the weight of the transition $t$ is
\[
\ov{\autwt}(t)=\Compswt_{t':\state{e}{m_{e}} \trans[b] \state{e'}{m_{e'}}} \autwt(t')
\]
where $t'$ ranges over transitions of $\Automaton_{\post^*}$.
This construction requires linear time in the number of b-transitions of
$\Automaton_{\post^*}$, i.e., $O(|\RSM|\cdot \theta_e)$.
It is straightforward to see that
\[
\Compswt_{\ov{\lambda} : \state{u}{m_u}\reach[\ov{S}]{*} q_f  } \ov{\autwt}(\lambda) =
\Compswt_{\lambda : \state{u}{m_u}\reach[S]{*} q_f } \autwt(\lambda)
\]
where $\ov{\lambda}$ and $\lambda$ range over accepting runs of
$\ov{\Automaton}_{\post^*}$ and $\Automaton_{\post^*}$ respectively, and $S$
refines $\ov{S}$.
Then, given a superconfiguration $\ov{c}=\conf{u}{\ov{S}}$, the extraction of
$\rsmdist(R,\ov{c})$ is done similarly to the configuration distance extraction,
in $O(|S|\cdot \kappa^2\cdot \theta_e^2)$ time.

\myparagraph{Node distances.}
For node distances, the task is to compute
$
\rsmdist(R, u) = \Compswt_{c=\conf{u}{S}} \rsmdist(R, c)
$
for every node $u$ of $\RSM$.
This reduces to treating the automaton $\Automaton_{\post^*}$ as a graph $G$,
and solving a traditional single-source distance problem, where the source set
contains all states with old marks (i.e., old states that appear in the initial
automaton $\Automaton$).
This requires $O(\Height \cdot |\Automaton_{\post^*}|)$ time for semirings of height
$\Height$.
An informal argument for these bounds is to observe that $G$ can be itself
encoded by a SESE RSM $\RSM_G$ with a single module, where the entry represents
the source set of nodes with old marks.
Then, running $\ConfDistBounded$ for the
corresponding semiring, we obtain a solution to the single-source distance
problem in the aforementioned times, as established in
Theorem~\ref{them:finite-height}.
Finally, computing same-context node distances requires $O(|\RSM|\cdot \theta)$
time in total (i.e., for all nodes).
Hence, regardless of the semiring, all node distances can be computed with no
overhead, i.e., within the time bounds required for constructing the respective
configuration automaton $\Automaton_{\post^*}$.
The following theorem summarizes the complexity bounds that we obtain for the
various distance extraction problems.

\begin{theorem}[Distance extraction]\label{them:weight_extraction}
Let $\RSM$ be an RSM over a semiring of height $\Height$
and $\Automaton$ an $\RSM$-automaton with $\kappa$ marks. 
After $O(\Height\cdot |\RSM|\cdot \theta_e\cdot \theta_x\cdot \kappa^3)$ preprocessing time
\begin{compactenum}
\item configuration and superconfiguration distance queries $\conf{u}{S}$ are
answered in $O(|S|\cdot \theta_e^2\cdot \kappa^2)$ time;
\item node distance queries are answered in $O(1)$ time.
\end{compactenum}
\end{theorem}

\subsection{Distances over Semirings with Small Domain}\label{subsec:semiring_dist_domain}

We now turn our attention to configuration and superconfiguration distance
extraction for the case of semirings with small domains $D$.
Such semirings express a range of important problems, with reachability being
the most well-known (expressed on the Boolean semiring with $|D|=2$).
We harness algorithmic advancements on the matrix-vector multiplication problem
and Four-Russians-style algorithms to obtain better bounds on the distance
extraction problem.

Recall that given a box $b$, the configuration automaton $\Automaton_{\post^*}$
has at most $(\theta_e\cdot \kappa)^2$ transitions labeled with $b$ .
Such transitions can be represented by a matrix
$A_{b}\in D^{(\theta_e\cdot \kappa) \times (\theta_e\cdot \kappa)}$.
Additionally, for every internal node $u$ we have one matrix
$A_u\in D^{(\kappa) \times (\theta_e\cdot \kappa)}$ that captures the weights of
all transitions of the form $\state{u}{m_u}\trans[\eps]\state{e}{m_e}$.
Then, answering a configuration distance query $\conf{u}{S}$ with
$S=b_1,\dots, b_{|S|}$ amounts to evaluating the expression
\begin{align}\label{eq:matrix_mult}
\mathbf{\one}_{\kappa} \cdot A_u\cdot A_{b_{1}} \cdots  A_{b_{|S|}} \cdot \mathbf{\one}_{\kappa\cdot \theta_e}^{\top}
\end{align}
where $\mathbf{\one}_{z}$ is a row vector of $\one$s and size $z$, $\cdot^\top$
denotes the transpose, and matrix multiplication is taken over the semiring.
The situation is similar in the case of superconfiguration distances, where we
have one matrix $A_{\Module, \Module'}$ for each pair of modules $\Module$, $\Module'$
such that $\Module$ invokes $\Module'$.

Evaluating equation~\eqref{eq:matrix_mult} from left to right (or right to left)
yields a sequence of matrix-vector multiplications. 
The following two theorems use the results of~\cite{Liberty09}
and~\cite{Williams07} on matrix-vector multiplications to provide a speedup on
the distance extraction problem when the semiring has constant size $|D|=O(1)$.

\begin{theorem}[Mailman's speedup~\cite{Liberty09}]\label{them:weight_extraction_mailman}
Let $\RSM$ be an RSM over a semiring of constant size, and
$\Automaton$ an $\RSM$-automaton with $\kappa$ marks. 
After $O(|\RSM|\cdot \theta_e\cdot \theta_x\cdot \kappa^3)$ preprocessing
time, configuration and superconfiguration distance queries $\conf{u}{S}$ are
answered in
$O\left(|S|\cdot \frac{ \theta_e^2\cdot \kappa^2}{\log (\theta_e \cdot
\kappa)}\right)$ time.
\end{theorem}

\begin{theorem}[Williams's speedup~\cite{Williams07}]\label{them:weight_extraction_williams}
Let $\RSM$ be an RSM over a semiring of size $|D|$, and
$\Automaton$ an $\RSM$-automaton with $\kappa$ marks. 
For any fixed $\eps>0$, let
$X=|\RSM|\cdot \theta_e\cdot \theta_x\cdot \kappa^3$ and
$Z=|\RSM|\cdot \kappa \cdot (\theta_e\cdot \kappa)^{1+\eps \log_2|D|}$.
After $O(\max(X,Z))$ preprocessing time, configuration and superconfiguration
distance queries $\conf{u}{S}$ are answered in
$O\left(|S|\cdot \frac{ \theta_e^2\cdot \kappa^2}{\eps^2\cdot \log^2 (\theta_e
\cdot \kappa)}\right)$ time.
\end{theorem}

Finally, using the Four-Russians technique for parsing on non-deterministic
automata~\cite{Myers92}, we obtain the following speedup for the case of
reachability.
We note that although the alphabet is not of constant size (i.e., the number of
boxes is generally non-constant) this poses no overhead, as long as comparing two
boxes for equality requires constant time (which is the case in the standard RAM
model).

\begin{theorem}[Four-Russians speedup~\cite{Myers92}]\label{them:weight_extraction_four_russians}
Let $\RSM$ be an RSM over a binary semiring, and $\Automaton$ an
$\RSM$-automaton with $\kappa$ marks. 
After $O(|\RSM|\cdot \theta_e\cdot \theta_x\cdot \kappa^3)$ preprocessing
time, configuration and superconfiguration distance queries $\conf{u}{S}$ are
answered in
$O\left(|\RSM|\cdot \theta_e\cdot \kappa^2\cdot \frac{ |S|}{\log
(|S|)}\right)$ time.
\end{theorem}

\subsection{A Speedup for Sparse RSMs}\label{subsec:sparse_dist}

We call an RSM $\RSM$ \emph{sparse} if there is a constant bound $r$ such that
for all modules $\Module_i$ we have $|\set{ Y_i(b) \mid b \in B_i}| \leq r$
i.e., every module invokes at most $r$ other modules (although $\Module_i$ can
have arbitrarily many boxes).
Typical call-graphs of most programs are very sparse, 
e.g., typical call graphs of thousands of nodes have average degree at most eight~\cite{Bhattacharya12,Qu15}.
Hence, an RSM modeling a typical program is expected to comprise thousands of
modules, while the average module invokes a small number of other modules.
Although this does not imply a constant bound on the number of invoked modules,
such an assumption provides a good theoretical basis for the analysis of typical
programs.

Our goal is to provide a speedup for extracting superconfiguration distances w.r.t.~a
sparse RSM.
This is achieved by an additional polynomial-time preprocessing, which then
allows to process a distance query in blocks of logarithmic size, and thus offers
a speedup of the same order.

Given an RSM $\RSM$ of $k$ modules and an integer $z$, there exist at most
$k\cdot r^z$ valid module sequences $\Module_1\dots, \Module_{z+1}$ which can
appear as a substring in a module sequence $\ov{S}$ which is refined by some
stack $S$.
Recall the definition of the matrices
$A_{\Module, \Module'}\in D^{(\theta_e\cdot \kappa) \times (\theta_e\cdot
\kappa)}$ from Section~\ref{subsec:semiring_dist_domain}.
For every valid sequence of $z+1$ modules $s=\Module_1\dots, \Module_{z+1}$, we
construct a matrix
$A_{s}=A_{\Module_1, \Module_2}\cdot A_{\Module_2,\Module_3}\cdot \dots
A_{\Module_z,\Module_{z+1}}$ in total time
\begin{align}\label{eq:matrix_complexity}
k\cdot (\theta_e\cdot \kappa)^{\omega} \sum_{i=1}^z r^i = O\left( |\RSM|\cdot  \theta_e^{\omega-1} \kappa^{\omega}\cdot r^z \right)
\end{align}
where $(\theta_e\cdot \kappa)^\omega=\Omega(\theta^2\cdot \kappa^2)$ is time
require to multiply two
$D^{(\theta_e\cdot \kappa) \times (\theta_e\cdot \kappa)}$ matrices
(currently $\omega\simeq 2.372$, due to~\cite{Williams12}).

Observe that as long as $z=O(\log|\RSM|)$, there are polynomially many such
sequences $s$, and thus each one can be indexed in $O(1)$ time on the standard
RAM model. Then a superconfiguration distance query $\tup{u,S}$ can be answered by
grouping $S$ in $\lceil\frac{|S|}{z}\rceil $ blocks of size $z$ each, and for each such block $s$
multiply with matrix $A_{s}$.

\begin{theorem}[Sparsity speedup]\label{them:weight_extraction_sparse}
Let $\RSM$ be a sparse RSM over a semiring of height $\Height$, 
and $\Automaton$ an $\RSM$-automaton with $\kappa$ marks. 
Let $X=\Height\cdot |\RSM|\cdot \theta_e\cdot \theta_x\cdot \kappa^3$,
and given an integer parameter $x = O(\poly|\RSM|)$, let
$Z=|\RSM|\cdot \theta_e^{\omega-1} \kappa^{\omega}\cdot x$.
After $O(\max(X,Z))$ preprocessing time, superconfiguration distance queries
$\conf{u}{S}$ are answered in
$O\left(|S|\cdot \left \lceil \frac{\theta_e^2\cdot \kappa^2}{\log
x}\right\rceil \right)$ time.
\end{theorem}

By varying the parameter $z$, Theorem~\ref{them:weight_extraction_sparse}
provides a tradeoff between preprocessing and query times.
Finally, the presented method can be combined with the preprocessing on
constant-size semirings of Section~\ref{subsec:semiring_dist_domain} which leads
to a $\Theta(\log z)$ factor improvement on the query times of
Theorem~\ref{them:weight_extraction_mailman},
Theorem~\ref{them:weight_extraction_williams}, and
Theorem~\ref{them:weight_extraction_four_russians}.


\section{Context-Bounded Reachability in Concurrent Recursive State Machines}
\label{sec:rsm_concurrent}

Context bounding, i.e., limiting the number of context switches considered
during state space exploration, is an effective technique for systematic 
analysis of concurrent programs.
The context-bounded reachability problem in concurrent pushdown systems has been
studied in~\cite{QadeerR05}.
In this section we phrase the context-bounded reachability problem over concurrent RSMs
and show that the procedure of~\cite{QadeerR05} using our algorithm
$\ConfDistBounded$ together with the results of the previous sections give a
better time complexity for the problem.
As the section follows closely the well-known framework of concurrent pushdown
systems~\cite{QadeerR05}, we keep the description brief.

\myparagraph{Concurrent RSMs.}
A \emph{concurrent RSM (CRSM)} $\CRSM$ is a collection of RSMs $\RSM_i$ equipped with a
finite set of \emph{global states} $G$ used for communication between the RSMs.
To this end, the semantics of RSMs is lifted to $\RSM_i$-configurations of the
form $\gconf{g}{u_i}{S_i}$, carrying an additional global state $g \in G$.
Then, a \emph{global configuration} of $\CRSM$ is a tuple
$\tup{g,\lconf{u_1}{S_1},\dots,\conf{u_n}{S_n}}$, where $\gconf{g}{u_i}{S_i}$
are configurations of $\RSM_i$, respectively.
The semantics of $\CRSM$ over global configurations is the standard interleaving
semantics, i.e., in each step some RSM $\RSM_i$ modifies the global state and
its local configuration, while the local configuration of every other RSM
remains unchanged.

\myparagraph{Context-bounded reachability.}
For a positive natural number $k$ and a fixed initial global configuration $c$,
the \emph{$k$-bounded reachability problem} asks for all global configurations
$c'$ such that there is a computation from $c$ to $c'$ that switches control
between RSMs at most $k-1$ times.

\myparagraph{An algorithm for context-bounded reachability.}
The procedure of~\cite{QadeerR05} for solving the $k$-bounded reachability
problem for \emph{concurrent pushdown systems (CPDSs)} systematically performs
$\post^*$ operations on the reachable configuration set of every constituent
PDS, while capturing all possible interleavings within $k$ context switches.
The $k$-bounded reachability problem for CRSMs can be solved with an almost
identical procedure, replacing the black-box invocations of the PDS reachability
algorithm of~\cite{Schwoon02} with our algorithm $\ConfDistBounded$.
However, using our algorithm for each $\post^*$ operation, we obtain a
complexity improvement over the method of~\cite{QadeerR05}.

\myparagraph{Key complexity improvement.}
The key advantage of our algorithm as compared to~\cite{QadeerR05} is as 
follows: in the algorithm of~\cite{QadeerR05}, in each iteration the 
configuration automata, used to represent the reachable configurations of each 
component RSM, grows by a cubic term; in contrast, replacing with our 
algorithm the configuration automata grows only by a linear term in each 
iteration.
This comes from the fact that in our configuration automata every state 
corresponds to a node of the RSM, whereas such strong correspondence does not 
hold for the configuration automata of~\cite{QadeerR05}.

\begin{restatable}{theorem}{concurrent}
\label{them:concurrent}
For a concurrent RSM $\CRSM$, and a bound $k$, the procedure of~\cite[Figure~2]{QadeerR05}
using $\ConfDistBounded$ for performing $\post^*$ operations
correctly solves the $k$-bounded reachability problem and requires
$O(|\CRSM|\cdot \theta_e^{||}\cdot  \theta_x^{||}\cdot  n^k\cdot |G|^{k+2})$ time.
\end{restatable}

\noindent
Compared to Theorem~\ref{them:concurrent}, solving the CRSM problem by
translation to a CPDS and using the algorithm of~\cite{QadeerR05} gives the
bound $O(|\CRSM|^5 \cdot \theta^{||~5}_x\cdot n^k\cdot |G|^k)$.
Conversely, solving the CPDS problem by translation to a CRSM and using our
algorithm gives an improvement by a factor $\Omega(|\CPDS|^3/|G|^2)$.
We refer to \appendixreftr{Appendix~\ref{app:discussion-concurrent}} for a
detailed discussion.


\section{Experimental Results}
\label{sec:evaluation}

In this section we empirically demonstrate the algorithmic improvements achieved
by our RSM-based algorithm over existing PDS-based algorithms on interprocedural
program analysis problems.
The main goal is to demonstrate the improvements in algorithmic ideas rather
than implementation details and engineering aspects.
In particular, we implemented our algorithm $\ConfDistBounded$ in a prototype
tool and compared its efficiency against jMoped~\cite{jmoped}, which implements
the algorithms of~\cite{Schwoon02,Reps05} and is a leading tool for the analysis
of weighted pushdown systems.
In all cases we used an explicit representation of data valuations on the nodes of
RSMs, as opposed to a symbolic semiring representation.
All experiments were run on a machine with an Intel Xeon CPU and a memory limit
of 80GB.
We first present our result on a synthetic example to verify the algorithmic
improvements on a constructed family, and then present results on real-world
benchmarks.

\subsection{A Family of Dense RSMs}\label{subsec:example}
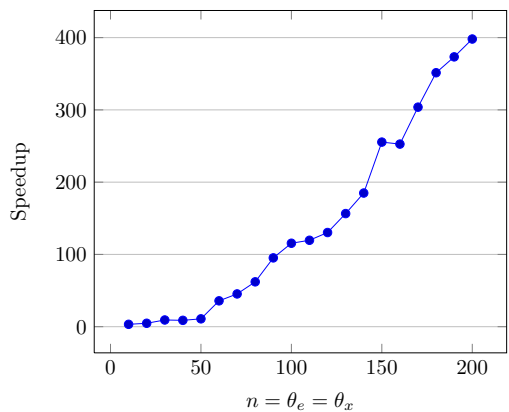
\begin{wrapfigure}{r}{0.54\textwidth}
  \centering
  \begin{tikzpicture}[scale=.8]
    \begin{axis}[
      xlabel={$n = \theta_e = \theta_x$},
      ylabel={Speedup},
      ymajorgrids=true,
      ]
      \addplot table [x=ex, y=speedup, col sep=tab] {graphics/scaling_dense.csv};
    \end{axis}
  \end{tikzpicture}
  \caption{Speedup of our algorithm over the algorithms of~\cite{Schwoon02,Reps05} implemented by jMoped on the RSM family $\RSM_n$.}
  \label{fig:scaling_dense}
\end{wrapfigure}


For our first experiments we constructed a family of dense RSMs that can be scaled 
in size.
The purpose of this experiments is to verify that (i)~our algorithm indeed
achieves a speedup over the algorithms of~\cite{Schwoon02,Reps05}, 
and (ii)~the speedup scales with the size of the input to ensure that
improvements on real-world benchmarks are not due to implementation details,
such as the used data types.
Let $\RSM_n$ be a single-module RSM that consists of $n$ entries and $n$ exits,
and a single box which makes a recursive call.
The transition relation is
$\delta = (\En \times (\Call \cup \Ex)) \cup (\Ret \times \Ex)$,
i.e., every entry node connects to every call and exit node, and every return
node connects to every exit node. 
Hence $|\RSM_n|=n^2$.
The transition weights are irrelevant, as we will focus on reachability.
The initial configuration automaton $\Automaton$ contains a single entry state.
We considered $\RSM_n$ with $n$ in the range from 10 to 200.
For each RSM, we used the standard translation to a PDS~\cite{AlurBEGRY05}, and
then applied our tool and jMoped to compute a configuration automaton that
represents $\post^{*}(\lang(\Automaton))$.
Figure~\ref{fig:scaling_dense} depicts the obtained speedup, which scales
linearly with $n$.
We have also experimented with other similar synthetic RSMs with different means 
of scaling; and in all cases the obtained speedups have the same qualitative behavior.
This confirms the theoretical algorithmic improvements of our algorithm on the
synthetic benchmarks. 

\subsection{Boolean Programs from SLAM/SDV}

\myparagraph{Benchmarks.}
For our second experiments we used the collection of Boolean programs distributed as
part of the SLAM/SDV project~\cite{BallBLKL10,BallR00}. 
These programs are the final abstractions in the verification of Windows device
drivers, and thus they represent RSMs obtained from real-world programs. 
From the Boolean programs we obtained RSMs where every node represents a control
location together with a valuation of Boolean variables, and call/entry and
exit/return nodes model the parameter passing between functions.
Thus, the RSMs are naturally multi-entry-multi-exit.
Overall we obtained 73 RSMs, which correspond to the largest Boolean programs
possible to handle explicitly.

\myparagraph{Evaluation.}
To ensure a fair performance comparison, we applied two preprocessing steps to
the benchmark RSMs.
\begin{compactitem}
\item First, to ensure that both tools compute the same result without any
potential unnecessary work, we restricted the state space of the RSMs to the
interprocedurally reachable states.
\item Second, to focus on the performance of interprocedural analysis, we eliminated
all internal nodes by computing the intraprocedural transitive closure within
every RSM module.
\end{compactitem}
The above two transformations ensure preprocessing steps like removal of
unreachable states and intraprocedural analysis is already done, and we compare
the interprocedural algorithmic aspects of the algorithms.
For each RSM, we used the standard translation to a PDS~\cite{AlurBEGRY05}, and
then applied our tool and jMoped to compute a configuration automaton that
represents $\post^{*}(\lang(\Automaton))$, where $\Automaton$ is an initial
configuration automaton that contains the entry states of the main module.
Table~\ref{tbl:bp-experiments} shows for every benchmark the number of RSM
transitions (Trans.), their ratio to nodes (D), the runtime for computing the intraprocedural transitive
closure (TC), the runtime of jMoped (jMop), the runtime of our tool (Ours), and
the speedup our tool achieved over jMoped (SpUp).

Out tool clearly outperforms jMoped on every benchmark, with speedups from 3.94
up to 28.48. 
The runtimes of our tool range from 0.13 to 33.96 seconds, while the runtimes of
jMoped range from 1.03 to 950.82 seconds.
Thus, our experiments show that also for real-world examples our algorithm
successfully exploits the structure of procedural programs preserved in RSMs.
This shows the potential of our algorithm for building program analysis tools.

Note that the benchmark RSMs are quite large, with millions of nodes and
transitions, which even a basic implementation of our algorithm handled quite 
efficiently.
Moreover, in our experiments we observed that our tool uses considerably less
memory than jMoped.
While we set 80GB as the memory limit, the peak memory consumption of jMoped 
was 72GB, whereas our tool solved all benchmarks with less than 32GB memory.

\begin{table}[tb]
  \newcommand{\tableheadings}{
    \bf \# &
    \bf Trans. &
    \bf D &
    \multicolumn{1}{c}{\bf TC} &
    \multicolumn{1}{c}{\bf jMop} &
    \multicolumn{1}{c}{\bf Ours} &
    \multicolumn{1}{c|}{\bf SpUp} \\}
  \newcommand{\tablecolspec}{|r|rr|cccc|}
  \scriptsize
    \begin{tabular}[t]{|r|rr|d{3.2}d{3.2}d{3.2}d{3.2}|}
      \hline
      \tableheadings
      \hline
       1 &    246,101 & ~1.9 &   1.18 &   1.10 &  0.28 & 3.94 \\
       2 &    216,021 & 0.8 &   0.70 &   1.03 &  0.26 & 3.96 \\
       3 &    593,041 & 1.5 &   1.05 &   2.05 &  0.49 & 4.19 \\
       4 &  1,043,217 & 1.2 &   3.01 &   4.67 &  1.11 & 4.20 \\
       5 &    329,088 & 1.4 &   1.41 &   1.43 &  0.34 & 4.24 \\
       6 & 10,281,149 & 3.0 &  11.36 &  52.00 & 10.61 & 4.90 \\
       7 &    908,092 & 1.7 &   2.04 &   3.31 &  0.65 & 5.08 \\
       8 &    969,388 & 2.2 &   2.00 &  33.71 &  6.60 & 5.11 \\
       9 &    298,126 & 1.5 &   0.68 &   1.31 &  0.25 & 5.23 \\
      10 &  1,780,776 & 1.3 &   5.82 &   6.44 &  1.20 & 5.35 \\
      11 &    163,853 & 1.4 &   0.33 &   1.03 &  0.19 & 5.35 \\
      12 &    205,608 & 1.0 &   0.50 &   4.62 &  0.86 & 5.36 \\
      13 & 28,568,561 & 1.7 &  23.21 & 102.54 & 18.82 & 5.45 \\
      14 & 21,911,277 & 1.8 &  15.79 &  80.41 & 14.64 & 5.49 \\
      15 &  2,453,881 & 1.5 &   4.54 &   9.57 &  1.72 & 5.55 \\
      16 &  5,833,574 & 1.8 &   6.97 &  21.14 &  3.80 & 5.56 \\
      17 &    332,768 & 0.8 &   0.77 &   2.28 &  0.41 & 5.59 \\
      18 &  1,782,697 & 1.3 &   5.79 &   6.70 &  1.20 & 5.60 \\
      19 &    246,127 & 1.9 &   1.31 &   1.36 &  0.24 & 5.63 \\
      20 & 21,648,560 & 1.8 &  15.50 &  79.45 & 14.01 & 5.67 \\
      21 &  7,033,834 & 2.1 &   8.23 &  23.97 &  4.21 & 5.70 \\
      22 & 28,944,391 & 1.7 &  24.26 & 105.00 & 18.15 & 5.78 \\
      23 &    464,004 & 1.7 &   0.75 &   2.17 &  0.37 & 5.83 \\
      24 &    424,916 & 1.6 &   1.20 &   2.94 &  0.49 & 5.96 \\
      25 & 22,186,326 & 1.6 &  17.77 &  63.27 & 10.56 & 5.99 \\
      26 & 11,719,007 & 5.2 &  20.36 &  52.29 &  8.55 & 6.11 \\
      27 &  2,989,001 & 1.4 &   3.55 &  11.04 &  1.80 & 6.12 \\
      28 &  1,952,647 & 1.3 &   3.83 &   7.98 &  1.30 & 6.13 \\
      29 &  7,970,359 & 3.2 &   4.04 &  30.16 &  4.70 & 6.42 \\
      30 &    682,435 & 2.1 &   2.14 &   4.88 &  0.76 & 6.42 \\
      31 &  9,480,799 & 4.9 &  17.23 &  44.34 &  6.77 & 6.55 \\
      32 &    845,867 & 2.4 &   1.59 &   3.22 &  0.48 & 6.67 \\
      33 &    953,420 & 3.1 &   1.22 &   4.51 &  0.67 & 6.77 \\
      34 &  1,205,731 & 2.0 &   3.31 &   4.68 &  0.68 & 6.84 \\
      35 &    754,270 & 1.7 &   4.25 &  22.28 &  3.23 & 6.90 \\
      36 &  1,463,749 & 2.0 &   2.38 &   6.10 &  0.88 & 6.95 \\
      37 &    434,884 & 5.8 &   6.85 &   1.90 &  0.27 & 7.10 \\
      \hline    
    \end{tabular}
    \hfill
    \begin{tabular}[t]{|r|rr|d{3.2}d{3.2}d{3.2}d{3.2}|}
      \hline
      \tableheadings
      \hline
      38 & 14,473,411 &   1.5 &   9.68 &  53.38 &  7.49 & 7.13 \\
      39 & 11,616,241 &   3.3 &  19.59 &  42.73 &  5.54 &  7.71 \\
      40 &    300,401 &   2.6 &   0.74 &   1.05 &  0.14 &  7.79 \\
      41 &  1,916,064 &   2.3 &   3.38 &  10.83 &  1.39 &  7.80 \\
      42 &    216,070 &   1.7 &   0.56 &   1.37 &  0.17 &  7.83 \\
      43 &  1,293,130 &   2.3 &   2.06 &   5.44 &  0.69 &  7.92 \\
      44 &  8,364,920 &   2.1 &   6.31 &  32.95 &  4.09 &  8.05 \\
      45 & 18,733,065 &   4.9 &  10.84 &  62.14 &  7.63 &  8.15 \\
      46 &  5,373,059 &   6.4 &   8.66 &  18.20 &  2.17 &  8.38 \\
      47 &  1,342,348 &   1.6 &   4.75 &   5.02 &  0.58 &  8.73 \\
      48 &    779,369 &   7.2 &   1.94 &   6.73 &  0.77 &  8.75 \\
      49 & 18,812,123 &   4.9 &   8.87 &  63.86 &  6.99 &  9.14 \\
      50 & 40,025,428 &   6.3 &  36.49 & 310.16 & 33.07 &  9.38 \\
      51 &  2,503,668 &  15.3 &  21.53 &  10.17 &  1.08 &  9.44 \\
      52 & 40,084,249 &   6.2 &  36.37 & 320.70 & 33.96 &  9.44 \\
      53 &  4,852,736 &   6.5 &   4.14 &  17.68 &  1.83 &  9.64 \\
      54 & 18,520,461 &   5.4 &   8.96 &  60.24 &  6.21 &  9.69 \\
      55 &  6,796,783 &   7.0 &   9.78 &  21.33 &  2.16 &  9.87 \\
      56 & 40,026,391 &   6.3 &  35.69 & 327.66 & 33.05 &  9.91 \\
      57 &    805,305 &   4.7 &   1.66 &   8.14 &  0.80 & 10.17 \\
      58 &  4,532,440 &  26.4 &   7.49 &  33.46 &  3.15 & 10.61 \\
      59 & 18,374,693 &   5.8 &   8.99 &  60.54 &  5.52 & 10.96 \\
      60 &  1,284,096 &   5.9 &   1.53 &  48.54 &  4.39 & 11.05 \\
      61 &  3,862,954 &   6.3 &   3.44 &  12.94 &  1.14 & 11.38 \\
      62 & 52,269,131 &   3.4 &  44.45 & 177.98 & 15.53 & 11.46 \\
      63 &    130,721 &   2.2 &   0.43 &   1.55 &  0.13 & 11.52 \\
      64 &    545,063 &  16.4 &   6.88 &   2.27 &  0.16 & 13.85 \\
      65 &    545,046 &  16.4 &   6.78 &   2.17 &  0.15 & 14.04 \\
      66 &    829,090 &  12.3 &   9.60 &   3.40 &  0.24 & 14.17 \\
      67 & 63,918,783 & 267.0 & 115.87 & 244.01 & 16.00 & 15.25 \\
      68 & 20,382,912 &   3.3 &  15.78 &  76.69 &  4.80 & 15.98 \\
      69 & 29,689,784 &   6.2 &  11.18 & 120.82 &  7.16 & 16.88 \\
      70 &  2,619,392 &   5.2 &   3.48 & 660.92 & 31.62 & 20.90 \\
      71 &  2,575,360 &   5.7 &   3.03 & 589.87 & 25.69 & 22.96 \\
      72 &  2,639,872 &   5.0 &   3.17 & 816.08 & 29.93 & 27.27 \\
      73 &  2,691,072 &   4.5 &   3.43 & 950.82 & 33.39 & 28.48 \\
         &            &       &        &        &       &       \\
      \hline
    \end{tabular}
  \caption{Comparison of our tool against jMoped. 
    Runtimes are given in seconds.
    The names of all benchmarks are given in \appendixreftr{Appendix~\ref{app:bp-names}}.}%
  \label{tbl:bp-experiments}
\end{table}


\subsection{Discussion}
In our experiments we compared the implementation of our algorithm with jMoped
on sequential RSM analysis in an explicit setting. 
While our algorithm can be made symbolic in a straightforward way\appendixref{
  (see Appendix~\ref{app:discussion-symbolic})}, a symbolic implementation and
efficiency for large symbolic domains involve significant engineering efforts.
Moreover, the main goal of our work is to compare the algorithmic improvements
over the existing approaches, which is best demonstrated in an explicit setting,
since in the explicit setting the improvements are algorithmic rather than due
to implementation details of symbolic data-structures.
Our experimental results show the potential of the new algorithmic ideas, and
investigating the applicability of them with a symbolic implementation is a
subject of future work.


\section{Related Work}

\myparagraph{Sequential setting.}
Pushdown systems are very well studied for interprocedural
analysis~\cite{Reps95,Sagiv96,Callahan86}.
While the most basic problem is reachability, the weighted pushdown systems
(i.e., pushdown systems enriched with semiring) can express several basic
dataflow properties, and other relevant problems in interprocedural program
analysis~\cite{Reps05,LalRB05,RepsLK07,LalR08}. 
Hence weighted pushdown systems have been studied in many different contexts,
such as~\cite{Sagiv96,Reps95,Horwitz95,Chaudhuri08}, and tools have been
developed, such as Moped~\cite{moped}, jMoped~\cite{jmoped}, and
WALi~\cite{WALi}.
The more convenient model of RSMs was introduced and studied
in~\cite{AlurBEGRY05}, which on the one hand explicitly models the function calls
and returns, and on the other hand specifies many natural parameters for
algorithmic analysis. 
In this work, we improve the fundamental algorithms for RSMs over finite-height
semirings, as compared to the bounds obtained by translating RSMs to pushdown
systems and applying the best-known bounds for the pushdown case.
Along with general RSMs, special cases of SESE RSMs have also been considered,
such as RSMs with constant treewidth, and only same context
queries~\cite{CIPG15} (i.e., computation of node distances between nodes of the
same module).
Our results apply to the general case of all RSMs and are not restricted to any
special types of queries.

\myparagraph{Concurrent setting.}
The problem of reachability in concurrent pushdown systems (or concurrent RSMs)
is again a fundamental problem in program analysis, which allows for the
interprocedural analysis in a concurrent setting.
However, the problem is undecidable~\cite{Ramalingam00}.
Motivated by practical problems, where bugs are discovered with few context
switches, the context-bounded reachability problem, where there can be at most
$k$ context switches have been considered for concurrent pushdown
systems~\cite{QadeerR05,MusuvathiQ07,MusuvathiQBBNN08,LalTKR08,LalR09} as well
as related models of asynchronous pushdown networks~\cite{Bouajjani05}.
We present a new algorithm for concurrent pushdown systems and concurrent RSMs
which improves the existing complexity when the size of the 
global component is small.


\section{Conclusion}\label{sec:conclusion}
In this work we consider RSMs, a fundamental model for interprocedural analysis,
with path properties expressed over finite-height semirings, that can express a
large class of properties for program analysis.
We present algorithms that improve the previous algorithms, both in the
sequential as well as in the concurrent setting.
Moreover, along with our algorithm, we present new methods to extract distances
from the data-structure (configuration automata) that the algorithm constructs.
We present a prototype implementation for sequential RSMs in an explicit setting
that provides significant improvements for real-world programs obtained from
SLAM/SDV benchmarks.
Our results show the potential of the new algorithmic ideas.
There are several interesting directions of future work.
A symbolic implementation is a direction for future work.
Another direction of future work is to explore the new algorithmic ideas in the
concurrent setting in practice.


\subsubsection*{Acknowledgments.}
\grants{}

\bibliographystyle{plain}
\bibliography{references}

\ifdefined\techreport
\clearpage
\appendix
\section{Proofs of Section~\ref{sec:rsm_sequential}}
\label{app:proofs-sequential}

\begin{proposition}\label{prop:run-marks}
  At all times, every run in the automaton contains at most one transition
  switching from the fresh mark to an old mark, and no transition switching from
  an old mark to the fresh mark.
  Furthermore, every accepting run has to end in a state with an old mark.
\end{proposition}

\lemcompleteness*
\begin{proof}
  First we show that for every initialized computation
  $\pi : \conf{u}{S} \Reach{*} \conf{u'}{S'}$ there is a run
  \[
    \overunderbraces
    {             & \br{2}{t_1}                                        &                    }
    { \lambda   : & \state{u'}{m_{u'}} \trans[\eps] & \state{e}{m_e}   & \reach[S']{*} q_f    }
    {             &                                 &                  &                  }
  \]
  accepting $\conf{u'}{S'}$ such that (1)~$\runwt(\lambda) \sleq \compwt(\pi)$,
  and (2)~$t_1$ was added to the worklist.
  We proceed by induction on the length of $\pi$.
  Part (2) of the induction hypothesis is used to argue that $t_1$ will be
  extracted with its final weight at some point in the algorithm.
  We do not explicitly prove this part below since it is an obvious consequence
  of the steps in the algorithm we refer to in order to prove part (1).
  
  As a base case, if $|\pi| = 0$ then $u=u'$, $S=S'$, and $\compwt(\pi) = \one$.
  Since $\conf{u}{S} \in \lang(\Automaton)$ there must be a $\Automaton$-run
  $\lambda$ accepting $\conf{u}{S}$, and since all transitions in the initial
  automaton have weight $\one$ we have $\runwt(\lambda) = \one$.

  For the induction step, if $|\pi| > 0$ there is a configuration
  $\conf{u_1}{S_1}$ such that
  \[ \pi : \underbrace{\conf{u}{S} \Reach{*} \conf{u_1}{S_1}}_{\pi_1} \Trans \conf{u'}{S'}. \]
  By applying the induction hypothesis to $\pi_1$ we obtain an accepting run
  \[
    \overunderbraces
    {             & \br{2}{t_1}                                        &                    }
    { \lambda_1 : & \state{u_1}{m_{u_1}} \trans[\eps] & \state{e}{m_e} & \reach[S_1]{*} q_f }
    {             &                                   & \br{2}{\lambda_1'}                  }
  \]
  such that
  \[ \runwt(\lambda_1) = \runwt(\lambda_1') \otimes \autwt(t_1) \sleq \compwt(\pi_1). \]
  Let $\Module_i$ be the module of $u_1$.
  We split cases according to the type of the last transition of $\pi$.

  \begin{enumerate}
  \item \textit{Internal transition:} $u' \in \In_i$, $\tup{u_1,u'} \in \delta_i$,
    and $S' = S_1$.
    We consider the iteration of the main loop where $t_1$ is extracted from
    $\WL$ with its final weight.
    Line~\ref{alg:rsm_post:relax1} relaxes the transition
    $t = \state{u'}{\freshMark} \trans[\eps] \state{e}{m_e}$ with
    $\autwt(t_1) \otimes \rsmwt_i(u_1,u')$ and hence
    \[ \autwt(t) \sleq \autwt(t_1) \otimes \rsmwt_i(u_1,u'). \]
    By combining $t$ and $\lambda_1'$ we obtain the accepting run
    \[ \lambda : \state{u'}{\freshMark} \trans[\eps] \state{e}{m_e} \reach[S']{*} q_f \]
    and we derive
    \begin{align*}
      \runwt(\lambda) &=     \runwt(\lambda_1') \otimes \autwt(t) \\
                      &\sleq \runwt(\lambda_1') \otimes \autwt(t_1) \otimes \rsmwt_i(u_1,u') \\
                      &\sleq \compwt(\pi_1) \otimes \rsmwt_i(u_1,u') \\
                      &=     \compwt(\pi).
    \end{align*}
  \item \textit{Call transition:} $u' = e' \in \En_{Y_i(b)}$ for some box
    $b \in B_i$, $\tup{u_1,\cnd{b}{e'}} \in \delta_i$, and $S' = b S_1$.

    Again, we consider the iteration of the main loop where $t_1$ is extracted from
    $\WL$ with its final weight.
    Line~\ref{alg:rsm_post:relax2} relaxes the transition
    $t = \state{e'}{\freshMark} \trans[b] \state{e}{m_e}$ with
    $\autwt(t_1) \otimes \rsmwt_i(u_1,\cnd{b}{e'})$ and hence
    \[ \autwt(t) \sleq \autwt(t_1) \otimes \rsmwt_i(u_1,\cnd{b}{e'}). \]
    By combining $t$ and $\lambda_1'$ we obtain the accepting run
    \[ \lambda : \state{e'}{\freshMark} \trans[\eps] \state{e'}{\freshMark} \trans[b] \state{e}{m_e} \reach[S']{*} q_f \]
    and we derive
    \begin{align*}
      \runwt(\lambda) &=     \runwt(\lambda_1') \otimes \autwt(t) \\
                      &\sleq \runwt(\lambda_1') \otimes \autwt(t_1) \otimes \rsmwt_i(u_1,\cnd{b}{e'}) \\
                      &\sleq \compwt(\pi_1) \otimes \rsmwt_i(u_1,\cnd{b}{e'}) \\
                      &=     \compwt(\pi).
    \end{align*}
  \item \textit{Return transition:} $u' = \rnd{b}{x} \in R_i$ for some
    $b \in B_i$ and $x \in \Ex_{Y_i(b)}$, $\tup{u_1,x} \in \delta_{Y_i(b)}$, and
    $S_1 = b S'$.
    First note that
    \[ \compwt(\pi_1) = \compwt(\pi). \]
    We consider the iteration of the main loop, where
    $t_\Automaton = \state{u_3}{m_{u_3}} \trans[\eps] \state{e}{m_e}$ is the
    transition extracted from $\WL$ such that the if-condition
    $\summary(\state{e}{m_e},x) \not\sleq \autwt(t_\Automaton)$ in
    line~\ref{alg:rsm_post:summary_if} holds for the last time.
    Then $\autwt(t_\Automaton) \sleq \autwt(t_1)$ and since
    line~\ref{alg:rsm_post:summarize} relaxes $\summary(\state{e}{m_e},x)$ with
    $\autwt(t_\Automaton)$ we have
    \[ \summary(\state{e}{m_e},x) \sleq \autwt(t_1). \]
    Now observe that we must have
    \[
      \overunderbraces
      {              &         \br{2}{t_2}                                                    &                    }
      { \lambda_1' : & \state{e}{m_e} \trans[b / v] & \state{e_2}{m_{e_2}} & \reach[S']{*} q_f.                    }
      {              &                                                                       & \br{2}{\lambda_1''}}
    \]
    We distinguish whether
    the transition $t_2$ already had weight $v$ in the current iteration of
    processing $t_\Automaton$ or not.
    If yes, then line~\ref{alg:rsm_post:relax3} in the current iteration, if no,
    then line~\ref{alg:rsm_post:relax4} in the later iteration where $t_2$ is
    extracted with weight $v$ from $\WL$, relaxes the transition
    $t = \state{\rnd{b}{x}}{\freshMark} \trans[\eps] \state{e_2}{m_{e_2}}$ with
    $v \otimes \summary(\state{e}{m_e},x)$ and hence
    \[ \autwt(t) \sleq v \otimes \summary(\state{e}{m_e},x). \]
    By combining $t$ and $\lambda_1''$ we obtain the accepting run
    \[ \lambda :\state{\rnd{b}{x}}{\freshMark} \trans[\eps] \state{e_2}{m_{e_2}} \reach[S']{*} q_f \]
    and we derive
    \begin{align*}
      \runwt(\lambda) &=     \runwt(\lambda_1'') \otimes \autwt(t) \\
                      &\sleq \runwt(\lambda_1'') \otimes v \otimes \summary(\state{e}{m_e},x) \\
                      &=     \runwt(\lambda_1') \otimes \summary(\state{e}{m_e},x) \\
                      &\sleq \runwt(\lambda_1') \otimes \autwt(t_1) \\
                      &\sleq \compwt(\pi_1) \\
                      &=     \compwt(\pi).
    \end{align*}
  \end{enumerate}
  In all cases we obtain the desired run $\lambda$ accepting $\conf{u'}{S'}$
  with $\runwt(\lambda) \sleq \compwt(\pi)$.
  Now the claim of the lemma follows, as
  \begin{align}\label{eq:completeness:proof}
    \Automaton_{\post^*}(c)
    =
    \Compswt_{\lambda \in \Runs(c) }\runwt(\lambda)
    \sleq
    \Compswt_{\pi \in \Comps(\lang(\Automaton),c)} \compwt(\pi)
    =
    \rsmdist(\lang(\Automaton), c)
  \end{align}
  were the inequality holds since, as shown above, the weight of any $\pi$ is
  bounded from below by the weight of a $\lambda$.\qed
\end{proof}



\begin{restatable}[Soundness invariants]{lemma}{sequentialsoundnessinv}\label{lem:soundness:inv}
  Algorithm $\ConfDistBounded$ maintains the following loop invariants:
  \begin{enumerate}[series=bound:inv,label={\textnormal{I\arabic*}}]
  \item\label{inv:sum} The function $\summary$ maintains sound summaries, i.e.,
    for every entry $e\in \En_i$ and exit $x\in \Ex_i$ of the same module
    $\Module_i$, and every box $b \in B_j$ with $Y_j(b) = i$, there exists a set
    $\Comps$ of computations $\pi : \conf{e}{b} \Reach{*} \conf{\rnd{b}{x}}{\eps}$
    such that $\Compswt(\Comps) = \summary(\state{e}{\freshMark},x)$.
  \item\label{inv:fresh} For every run
    $\lambda : \state{u_2}{\freshMark} \reach[S]{*} \state{u_1}{\freshMark}$,
    there exists a set $\Comps$ of computations
    $\pi : \conf{u_1}{\eps} \Reach{*} \conf{u_2}{S}$, such that
    $\Compswt(\Comps) = \runwt(\lambda)$.
  \item\label{inv:sound} For every run $\lambda$ accepting a configuration $c$
    there exists a set $\Comps$ of initialized computations ending in $c$, such
    that $\Compswt(\Comps) = \runwt(\lambda)$.
  \end{enumerate}
\end{restatable}


\begin{proof}
  Note that every accepting run has to end in a final state with an old mark
  since we do not add final states to the automaton.
  Moreover, we recall that Proposition~\ref{prop:run-marks} implies that every
  run contains at most one transition switching from the fresh mark to an old
  mark, and no transition switching from an old mark to the fresh mark.
  
  Initially, the invariants hold due to the initialization steps of the
  algorithm.
  
  Now we need to show that~\ref{inv:sum} is preserved by updates to the summary
  function in line~\ref{alg:rsm_post:summarize}, and that ~\ref{inv:fresh}
  and~\ref{inv:sound} are preserved by all possible relaxations performed in
  line~\ref{alg:rsm_post:relax1}, \ref{alg:rsm_post:relax2},
  \ref{alg:rsm_post:relax3}, and \ref{alg:rsm_post:relax4}.
  
  Since all cases follow the similar pattern of applying the invariants to
  sub-runs and combining the obtained sets of computations suitably, we only
  give the details of one case for~\ref{inv:sound}.
  
  Consider the run
  \[
  \overunderbraces
  {           &                                        & \br{3}{t}                                                           &                                        }
  { \lambda : & \state{u_2}{\freshMark} \reach[S_2]{*} & \state{e'}{\freshMark} & \trans[b / v_t] & \state{e}{m_e} & \reach[S_1]{*} q_f }
  {           & \br{2}{\lambda_2}                                               &                    & \br{2}{\lambda_1}                                              }
  \]
  with $m_e \neq \freshMark$ and an iteration of the main loop where transition
  $t_\Automaton = \state{u}{m_u} \trans[\eps] \state{e}{m_e}$ is extracted from
  $\WL$ and transition $t$ is relaxed in line~\ref{alg:rsm_post:relax2} due to a
  call transition $t_\RSM = \tup{u,\cnd{b}{e'}}$.
  
  We apply~\ref{inv:sound} to $t_\Automaton,\lambda_1$ to obtain a set of
  initialized computations ending in $\conf{u}{S_1}$. 
  We extended each computation by $t_\RSM$ to obtain the set $\Comps_1'$ of
  initialized computations ending in $\conf{e'}{bS_1}$.
  We apply~\ref{inv:sound} to $t,\lambda_1$ to obtain a set $\Comps_1''$ of
  initialized computations ending in $\conf{e'}{bS_1}$.
  We apply~\ref{inv:fresh} to $\lambda_2$ to obtain a set of computations from
  $\conf{e'}{\eps}$ to $\conf{u_2}{S_2}$.
  We lift the stack of each computation by $bS_1$ to obtain the set $\Comps_2$
  of computations from $\conf{e'}{bS_1}$ to $\conf{u_2}{S_2bS_1}$.
  Let $\Comps$ be the set of computations obtained by combining every
  computation in $\Comps_1' \cup \Comps_1''$ with every computation in
  $\Comps_2$.
  Then $\Comps$ is the desired set of initialized computations ending in
  $\conf{u_2}{S_2bS_1}$, such that $\Compswt(\Comps) = \runwt(\lambda)$ after
  the relaxation of $t$.\qed
\end{proof}


\lemsoundness*
\begin{proof}
  Conversely to
  the inequality of equation~\eqref{eq:completeness:proof} in the proof of
  Lemma~\ref{lem:completeness}we derive
  \begin{align}
    \rsmdist(\lang(\Automaton), c)
    =
    \Compswt_{\pi \in \Comps(\lang(\Automaton),c)} \compwt(\pi)
    \sleq
    \Compswt_{\lambda \in \Runs(c) }\runwt(\lambda)
    =
    \Automaton_{\post^*}(c)
  \end{align}
  were the inequality holds since, by invariant~\ref{inv:sound} from Lemma~\ref{lem:soundness:inv}, the weight of any $\lambda$ is
  bounded from below by the weight of a $\Comps \subseteq \Comps(\lang(\Automaton),c)$.\qed
\end{proof}


\lemcomplexity*
\begin{proof}
We bound the number of times that each loop will be executed.
\begin{compactenum}
\item For a given node $u$, line~\ref{alg:rsm_post:extract_wl} will be
executed at most $\Height \cdot \theta_e\cdot (\kappa+1)^2$ times.
We denote by $\Outgoing_i(u)$, $\Outgoing_c(u)$, $\Outgoing_x(u)$ the number
of internal, call, and exit transitions from node $u$
in $\delta_i$ of module $\Module_i$.
For every iteration of line~\ref{alg:rsm_post:while}, the following upper
bounds on each inner loop are straightforward:
\begin{compactenum}
\item Line~\ref{alg:rsm_post:internal_trans_loop}:
$\Outgoing_i(u)$ times.
\item Line~\ref{alg:rsm_post:call_trans_loop}:
$\Outgoing_c(u)$ times.
\item Line~\ref{alg:rsm_post:exit_trans_loop}:
$\Outgoing_x(u)$ times.
\end{compactenum}
Hence for a given pair $\state{u}{m}$ the algorithm spends
$O(\Height\cdot \theta_e\cdot (\Outgoing_i(u)+\Outgoing_c(u) + \Outgoing_x(u))\cdot \kappa^2)$
time in the above loops, and summing over all $u$ we obtain
$O(\Height\cdot \theta_e\cdot |\RSM|\cdot \kappa^2)$ time.
\item Given a pair of a state and an exit $(\state{e}{m_e},x)$,
Line~\ref{alg:rsm_post:summary_if} will hold true at most $\Height$ times.
Summing over all possible such pairs, we obtain that Line~\ref{alg:rsm_post:entry_exit_trans_loop} will be executed
$O(\Height \cdot |\Call| \cdot  \theta_e\cdot \theta_x\cdot  \kappa^3)$ times in total.
\item Finally, line~\ref{alg:rsm_post:relax4} will be executed $O(\Height \cdot |\Call|\cdot \theta_e\cdot \theta_x\cdot  \kappa^2)$ times in total,
since the total number of different edges of the form $\state{e}{m_e} \trans[b] \state{e'}{m_{e'}}$ added in the worklist
is bounded by the number of call nodes that were used in Line~\ref{alg:rsm_post:call_trans_loop}, times the maximum number of entries and exits in any module of the RSM.
\end{compactenum}
The desired result follows.\qed
\end{proof}



\section{Symbolic Extensions}
\label{app:discussion-symbolic}

Note that in our framework we deal with explicit RSMs, and our
$\ConfDistBounded$ algorithm is also explicit.
However, our results carry over to symbolic extensions of RSMs, similar to
symbolic PDS~\cite{HBMC:Procedural}.
For the symbolic extension we describe the symbolic extension of the model and
our algorithm.
\begin{compactitem}
\item {\em Symbolic model.} 
We consider the RSM to represent the control-flow structure of a program
(represented explicitly), and the semiring capturing valuations on the variables
(represented symbolically).
The symbolic semiring operations express the value changes along the program 
execution.
In general the MEME RSMs can represent explicitly a combination of 
control-flow and some Boolean variables.

\item {\em Symbolic algorithm.} 
We observe that our algorithm for computation on the semiring uses the basic 
semiring operations.
Hence our algorithm can be straightforwardly made symbolic on the semiring, and 
is only explicit on the RSM structure.
\end{compactitem}

\section{Proofs of Section~\ref{sec:rsm_concurrent}}
\label{app:proofs-concurrent}

\concurrent*
\begin{proof}
  The algorithm of~\cite[Figure~3]{QadeerR05} for concurrent RSMs basically
  calls algorithm for sequential RSMs as a black-box procedure.
  Using our algorithm $\ConfDistBounded$ we obtain an algorithm for concurrent
  RSM, and the correctness follows from~\cite{QadeerR05} and
  Theorem~\ref{them:finite-height}. 

  We now sketch the complexity analysis.
  By Theorem~\ref{them:finite-height}, every execution of the algorithm
  $\ConfDistBounded$ increases the number of marks of the input configuration
  automaton by $1$.
  In the $i$-th iteration of~\cite[Figure~3]{QadeerR05} the algorithm will
  perform a $\post^*$ operation on a configuration automaton of $i$ marks, which
  by Theorem~\ref{them:finite-height} will require
  $O(|\RSM|\cdot |G|^2\cdot \theta_e\cdot \theta_x \cdot i^3)$ time.
  Each such iteration will spawn $n\cdot |G|$ iterations in the inner loop
  of~\cite[Figure~3]{QadeerR05}, one for each component RSM (among $n$
  components) and state of the global component (among $|G|$ possible states).
  Then the total time is (up to constant factors)
  \[
    |\RSM|\cdot |G|^2\cdot \theta_e\cdot \theta_x \cdot \sum_{i=1}^{k}{i^3\cdot
      (n\cdot |G|)^i} = O\left(|\RSM|\cdot |G|^2\cdot \theta_e\cdot \theta_x\cdot
      (n\cdot |G|)^k\right)
  \]
  The desired result follows.\qed
\end{proof}


\section{Comparison with existing work on $k$-bounded reachability}
\label{app:discussion-concurrent}

We compare our results for CPDSs, CRSMs, and the related model of asynchronous
pushdown networks (APNs).

\myparagraph{Comparison for CPDSs.}
As shown in~\cite{QadeerR05}, the $k$-bounded reachability problem in a
CPDS $\CPDS$ can be solved in time
\[
O(|\CPDS|^5\cdot n^k\cdot |G|^k).
\]
The bisimulation relation of~\cite[Theorem~1]{AlurBEGRY05} between PDSs and RSMs has a straightforward 
extension to CPDSs and CRSMs.
In particular:
\begin{compactenum}
\item Given a CPDS $\CPDS$ with $n$ components and global set $G$, the
$k$-bounded reachability problem for $\CPDS$ can be reduced to the $k$-bounded
reachability problem for a CRSM $\CRSM$ with $n$ components and global set
$G$. 
Additionally,
\[
|\CRSM|=\Theta(|\CPDS|);\quad  \theta_e^{||}=\Theta(|\Gamma|)=O( |\CPDS|); \quad  \theta_x^{||}=\Theta(1)
\]
\item Given a CRSM $\CRSM$ with $n$ components and global set $G$, the
$k$-bounded reachability problem for $\CRSM$ can be reduced to the $k$-bounded
reachability problem for a CPDS $\CPDS$ with $n$ components and global set
$G$. 
Additionally,
\[
|\CPDS|=\Theta(|\CRSM| \cdot \theta_x^{||});
\]
\end{compactenum}
It follows that our approach can solve the $k$-bounded reachability problem of a
CPDS $\CPDS$ in time 
\[
O(|\CPDS|^2\cdot n^k\cdot |G|^{k+2}).
\] 
Note that typically $k$ is very small, (e.g. 
$k=2$ in~\cite{Qadeer04}, $k=3$ in~\cite{Suwimonteerabuth08}).
However, in real applications the size of $G$ is typically smaller than 
$|\CPDS|$, e.g., when $G$ encodes only the synchronization variables among threads.
Our algorithm gives an improvement by a factor $\Omega(|\CPDS|^3/|G|^2)$.

\myparagraph{Comparison for CRSMs.} 
The naive upper bound for the $k$-bounded reachability problem of a CRSM $\CRSM$ 
obtained using a modification of the standard method of~\cite{AlurBEGRY05} to reduce it to the
CPDS case, and then apply the algorithm of~\cite{QadeerR05}, is 
$O(|\CRSM|^5 \cdot \theta^{||~5}_x\cdot n^k\cdot |G|^k)$. 
In contrast, our bound is 
$O(|\CRSM|\cdot \theta_e^{||}\cdot  \theta_x^{||}\cdot  n^k\cdot |G|^{k+2})$. 

\myparagraph{Comparison for APNs.} 
The problem of $k$-bounded reachability has also been studied in the closely
related model of \emph{asynchronous pushdown networks (APNs)}~\cite{Bouajjani05}.
Informally, the main difference of an APN from a CPDS is that in the former
case, the stacks have an additional set of local control states, different from
the common global finite control $G$.
Hence APNs are more general than CPDS.
As shown in~\cite{Bouajjani05}, the $k$-bounded reachability problem for an APN
$\APN$ of $n$ components can be solved essentially in time
$O(n^k\cdot |G|^k +n\cdot |G|^{k+2}\cdot |\APN|^2\cdot |P| )$, where $P$ is the
set of local control states.
Since APNs are more general than CPDSs, our previous analysis implies that the
algorithm of~\cite{Bouajjani05} can be used to solve the $k$-bounded
reachability problem for a CRSM $\CRSM$ in time
$O(n^k\cdot |G|^k +n\cdot |G|^{k+2}\cdot |\CRSM|^2\cdot \theta^{||}_x)$.
This time is incomparable with what we obtain from
Theorem~\ref{them:concurrent}. 
When the number of entries $\theta^{||}_e$ and the number of components $n$ is
constant, the algorithm presented in this work has a better complexity.


\section{Names of the Boolean programs used as benchmarks}
\label{app:bp-names}

All listed benchmarks belong to the ``bebop-itp'' collection.

\begin{enumerate}[noitemsep]
\item src\_7600\_general\_toaster\_kmdf\_filter\_generic\_\_InvalidReqAccess
\item src\_7600\_general\_toaster\_wdm\_filter\_devupper\_\_PnpSurpriseRemove
\item src\_7600\_general\_event\_wdm\_\_MarkIrpPending2
\item src\_7600\_storage\_sfloppy\_\_PagedCode
\item src\_7600\_input\_kbfiltr\_sys\_\_InvalidReqAccess
\item src\_7600\_general\_toaster\_wdm\_toastmon\_\_PendedCompletedRequest
\item src\_7600\_general\_toaster\_wdm\_filter\_devupper\_\_CriticalRegions
\item src\_7600\_network\_ndis\_athwifi\_driver\_atheros\_\_Irql\_SendRcv\_Function
\item src\_7600\_general\_event\_wdm\_\_IrpProcessingComplete
\item src\_7600\_general\_toaster\_wdm\_func\_featured1\_\_WmiForward
\item src\_7600\_general\_toaster\_wdm\_func\_featured1\_\_IrqlKeWaitForSingleObject
\item src\_7600\_network\_ndis\_athwifi\_driver\_atheros\_\_Irql\_Timer\_Function
\item src\_7600\_general\_toaster\_wdm\_func\_featured1\_\_IrqlIoPassive3
\item src\_7600\_general\_toaster\_wdm\_func\_featured2\_\_IrqlIoApcLte
\item src\_7600\_general\_toaster\_wdm\_func\_featured2\_\_WmiForward
\item src\_7600\_general\_ioctl\_kmdf\_sys\_\_InitFreeDeviceCreateType4
\item src\_7600\_general\_toaster\_wdm\_func\_incomplete2\_\_IrqlReturn
\item src\_7600\_general\_pcidrv\_wdm\_hw\_\_WmiForward
\item src\_7600\_input\_moufiltr\_\_InvalidReqAccess
\item src\_7600\_general\_toaster\_wdm\_func\_featured2\_\_IrqlIoPassive3
\item src\_7600\_general\_toaster\_kmdf\_filter\_sideband\_\_ControlDeviceInitAPI
\item src\_7600\_general\_toaster\_wdm\_func\_featured1\_\_IrqlIoApcLte
\item src\_7600\_general\_toaster\_wdm\_func\_featured2\_\_PnpSurpriseRemove
\item src\_7600\_general\_amcc5933\_sys\_\_KmdfIrql
\item src\_7600\_general\_toaster\_wdm\_func\_incomplete1\_\_IrpProcessingComplete
\item src\_7600\_general\_toaster\_wdm\_filter\_devupper\_\_PendedCompletedRequest
\item src\_7600\_general\_toaster\_wdm\_func\_incomplete1\_\_TargetRelationNeedsRef
\item src\_7600\_general\_toaster\_wdm\_func\_incomplete2\_\_TargetRelationNeedsRef
\item src\_7600\_general\_toaster\_wdm\_func\_featured1\_\_PnpSurpriseRemove
\item src\_7600\_bth\_bthecho\_bthcli\_sys\_\_RequestFormattedValid
\item src\_7600\_general\_toaster\_wdm\_filter\_buslower\_\_PendedCompletedRequest
\item src\_7600\_input\_hiddigi\_wacompen\_\_MarkIrpPending
\item src\_7600\_general\_toaster\_wdm\_bus\_\_PnpSurpriseRemove
\item src\_7600\_general\_toaster\_kmdf\_bus\_static\_\_PdoInitFreeDeviceCreateType4
\item src\_7600\_network\_ndis\_athwifi\_driver\_atheros\_\_Irql\_IrqlSetting\_Function
\item src\_7600\_serial\_serenum\_\_IrqlIoApcLte
\item src\_7600\_hid\_hidusbfx2\_sys\_\_SyncReqSend2
\item src\_7600\_general\_toaster\_wdm\_toastmon\_\_IrpProcessingComplete
\item src\_7600\_general\_toaster\_kmdf\_filter\_sideband\_\_ControlDeviceDeleted
\item src\_7600\_general\_cancel\_startio\_\_MarkIrpPending2
\item src\_7600\_general\_toaster\_wdm\_func\_featured2\_\_IrqlKeSetEvent
\item src\_7600\_general\_toaster\_wdm\_func\_incomplete2\_\_IrqlKeSetEvent
\item src\_7600\_serial\_serenum\_\_IrqlIoPassive3
\item src\_7600\_general\_toaster\_wdm\_filter\_devlower\_\_IrpProcessingComplete
\item src\_7600\_general\_toaster\_wdm\_func\_featured1\_\_IrqlExAllocatePool
\item src\_7600\_general\_toaster\_wdm\_func\_featured2\_\_CriticalRegions
\item src\_7600\_general\_toaster\_wdm\_bus\_\_MarkIrpPending2
\item src\_7600\_storage\_class\_cdrom\_\_WdfSpinlockRelease
\item src\_7600\_general\_toaster\_wdm\_func\_featured2\_\_IrqlKeWaitForSingleObject
\item src\_7600\_network\_ndis\_xframeii\_sys\_ndis6\_\_MandatoryOid
\item src\_7600\_general\_toaster\_kmdf\_filter\_sideband\_\_DeviceInitAllocate
\item src\_7600\_network\_ndis\_xframeii\_sys\_ndis6\_\_NdisStallExecution\_Delay
\item src\_7600\_bth\_bthecho\_bthsrv\_sys\_\_InvalidReqAccessLocal
\item src\_7600\_general\_toaster\_wdm\_func\_featured2\_\_IrqlKeApcLte2
\item src\_7600\_general\_toaster\_wdm\_bus\_\_IrqlIoPassive3
\item src\_7600\_network\_ndis\_xframeii\_sys\_ndis6\_\_SpinlockRelease
\item src\_7600\_storage\_filters\_diskperf\_\_PnpIrpCompletion
\item src\_7600\_storage\_filters\_diskperf\_\_TargetRelationNeedsRef
\item src\_7600\_general\_toaster\_wdm\_func\_featured2\_\_IrqlZwPassive
\item src\_7600\_hid\_hidusbfx2\_hidmapper\_\_ForwardedAtBadIrql
\item src\_7600\_general\_pcidrv\_wdm\_hw\_\_DoubleCompletion
\item src\_7600\_serial\_serenum\_\_MarkIrpPending2
\item src\_7600\_general\_toaster\_wdm\_func\_incomplete2\_\_MarkIrpPending2
\item src\_7600\_input\_moufiltr\_\_SyncReqSend2
\item src\_7600\_general\_toaster\_kmdf\_filter\_generic\_\_SyncReqSend2
\item src\_7600\_input\_kbfiltr\_sys\_\_SyncReqSend2
\item src\_7600\_input\_hiddigi\_wacompen\_\_SpinLock
\item src\_7600\_general\_toaster\_wdm\_filter\_devupper\_\_IrpProcessingComplete
\item src\_7600\_general\_pcidrv\_wdm\_hw\_\_IrqlIoPassive1
\item src\_7600\_general\_toaster\_wdm\_func\_incomplete1\_\_ForwardedAtBadIrql
\item src\_7600\_general\_toaster\_wdm\_filter\_devlower\_\_ForwardedAtBadIrql
\item src\_7600\_general\_toaster\_wdm\_toastmon\_\_ForwardedAtBadIrql
\item src\_7600\_general\_toaster\_wdm\_filter\_devupper\_\_ForwardedAtBadIrql
\end{enumerate}


\fi

\end{document}